\documentclass[11pt]{article}
\usepackage{amsmath,graphicx}
\usepackage{hyperref}
\usepackage{xcolor,enumerate}
\usepackage[noend,linesnumbered,ruled,lined]{algorithm2e}
\SetAlFnt{\scriptsize}
\usepackage{comment,subfigure}
\usepackage{enumitem}
\usepackage{lineno}
\usepackage{soul}

\newenvironment{proof}{\noindent {\bf Proof.}}{}

\newtheorem{theorem}{Theorem}[section]
\newtheorem{lemma}{Lemma}[section]

\newcommand{\cA}{{\cal A}}
\newcommand{\dpn}{Dispersion}

\newcommand{\ssl}{{\sc Send\_Signal}}
\newcommand{\lsl}{{\sc Learn\_Signal}}
\newcommand{\actv}{{\sc Active}}
\newcommand{\master}{{\sc Master}}
\newcommand{\follower}{{\sc Follower}}

\newcommand{\pb}{{\sc ProcessBit}}

\begin{document}
\bibliographystyle{plain}

\title{{\bf Collaborative Dispersion by Silent Robots}}

\author{
Barun Gorain \footnotemark[1]
\and
Partha Sarathi Mandal \footnotemark[2]
\and
Kaushik Mondal  \footnotemark[3]
\and
Supantha Pandit \footnotemark[4]
}

\date{ }
\maketitle
\def\thefootnote{\fnsymbol{footnote}}

\footnotetext[1]{
\noindent
Department of Electrical Engineering and Computer Science, Indian Institute of Technology Bhilai, India.
{\tt barun@iitbhilai.ac.in}}

\footnotetext[2]{
\noindent
Department of Mathematics, Indian Institute of Technology Guwahati, India
{\tt psm@iitg.ac.in}}

\footnotetext[3]{
\noindent
Department of Mathematics, Indian Institute of Technology Guwahati, India
{\tt kaushik.mondal@iitropar.ac.in}}

\footnotetext[4]{
\noindent
Dhirubhai Ambani Institute of Information and Communication Technology, Gandhinagar, Gujrat, India.
{\tt pantha.pandit@gmail.com}
}

\maketitle              
\begin{abstract}

In the dispersion problem, a set of $k$ co-located mobile robots must relocate themselves in distinct nodes of an unknown network. The network is modeled as an anonymous graph $G=(V,E)$, where the nodes of the graph are not labeled. The edges incident to a node $v$ with degree $d$ are labeled with port numbers in the range $0,1, \cdots, d-1$ at $v$. The robots have unique ids in the range $[0,L]$, where $L \ge k$, and are initially placed at a source node $s$. Each robot knows only its own id but does not know the ids of the other robots or the values of $L,k$.
The task of dispersion was traditionally achieved with the assumption of two types of communication abilities: (a) when some robots are at the same node, they can communicate by exchanging messages between them (b) any two robots in the network can exchange messages between them.

In this paper, we ask whether this ability of communication  among  co-located robots is necessary to achieve dispersion. We show that even if the ability of communication is not available, the task of dispersion by a set of mobile robots can be achieved in a much weaker model where  a robot  at a node $v$ has the access of following very restricted  information at the beginning of any round: (1) am I alone at $v$? (2) the number of robots at $v$ increased or decreased compare to the previous round?

We propose a deterministic algorithm that achieves dispersion on any given graph $G=(V,E)$ in time $O\left( k\log L+k^2 \log \Delta\right)$, where $\Delta$ is the maximum degree of a node in $G$. Each robot uses $O(\log L+ \log \Delta)$ additional memory. We also prove that the task of dispersion cannot be achieved by a set of mobile robots with $o(\log L + \log \Delta)$ additional memory.
\end{abstract}

\section{Introduction}
\subsection{Background}
The \dpn~problem in a graph using mobile robots became popular in very recent times.  In this problem, a set of $k$ mobile robots, starting from one or multiple source nodes, must relocate themselves in the nodes of the graph so that no two robots are placed on a single node. This problem was first introduced by Augustine and Moses Jr. \cite{AugustineM18}. In the last few years, this problem got attention from various researchers and has been studied over various models.
This problem has several practical applications. The most prominent application is charging self-driving electric cars in charging stations \cite{AugustineM18}.
It is assumed that charging a car is a time-consuming and costly task than relocating the car to a nearby free charging station. So it is better to spread the cars such that each charging station gets one at any time instead of a long queue in a single station.
The \dpn~problem is closely related to several problems on a graph network such as exploration  \cite{BrassCGX11,BrassVX14,DereniowskiDKPU15}, scattering \cite{BarriereFBS11,ElorB11,ShibataMOKM16}, load balancing, etc. In the previous studies of dispersion, it is assumed that if two robots are co-located at the same node, they can communicate and exchange any amount of information. It enables the robots to learn the number of co-located robots in the node, their ids, the previous histories, etc.

\subsection{Motivation and Problem Definition}


Our work is motivated by the recent work on gathering by Bouchard et al. \cite{BouchardDP20}. In the problem of gathering \cite{Alpern2003,Pelc12}, a set of mobile robots, starting from different nodes of an unknown graph, must meet at a node and declare that they all met. In all the prior works related to gathering, the mobile robots are assumed to have the capability of communication: any two robots can communicate if they are co-located at a node. Bouchard et al. \cite{BouchardDP20} asked the following fundamental question: whether the capability of communication between co-located robots is necessary for gathering? They show that  gathering can be achieved without communication by a set of co-located mobile robots. Here, it is assumed that a robot at any node can see how many robots are co-located with it in any round.

The task in the problem of dispersion is the opposite of gathering. Here, a set of co-located mobile robots must be relocated to different nodes of the graph. Similar to the problem of gathering, in all prior works in dispersion, the capability of communication between co-located robots is assumed. Therefore, it is  natural to ask whether this communication capability is necessary to solve dispersion.

\subsection{ The Model}

Let $G=(V,E)$ be a connected graph with $n$ nodes. The nodes of the graph are anonymous but the edges incident to a node $v$ of degree $d$ are labeled arbitrarily by unique port numbers $0,1, \cdots,d-1$. Thus, every edge in $E$ is associated with two independent port numbers, one corresponding to each of its end nodes. Let $s$ be a specified source node in $V$ and initially $k\le n$ mobile robots are placed at $s$. Each mobile robot has a unique integer id represented as a binary string in the range $[0,L]$, $k-1 \leq L$.

A mobile robot knows its own id,  but does not know the ids of the other robots or the values of $L$ and $k$. The robots move in synchronous rounds, and at most one edge can be traversed by a robot in every round. Each round is divided into two different stages. In the first stage, each robot at a node $v$ does any amount of local computations. In the second stage, a robot moves along one of the edges incidents to $v$ or stays at $v$.

The robots are silent: there is no means of communication between any two robots in the graph. A robot $M$ at a node $v$ has access to the following two local information at the beginning of a round: (1) am I alone at $v$? (2) the number of co-located robots increased or decreased compared to the previous round.

The robots use three binary variables named $alone$, $increase$ and $decrease$ in order to store the above mentioned local information. In any round $t$, If there is only one robot at $v$, then $alone=true$, else $alone=false$. If a robot $M$ decides to stay at a node $v$ in the $t$-th round, and if the number of robots that left $v$ in the $t$-th round is more than the number of robots that entered $v$, then $decrease=true$ for $M$ at the beginning of the $(t+1)$-th round. If a robot $M$ decides to stay at a node $v$ in the $t$-th round, and if the number of robots that left $v$ in the $t$-th round is less than the number of robots that entered $v$, then $increase=true$ for $M$ at the beginning of the $(t+1)$-th round. Otherwise, both $increase$ and $decrease$ are $false$ at the beginning of the $(t+1)$-th round.

Whenever a robot $M$ entered a node $v$ in some round $t$, using a port from another node, it learns the incoming port through which it reaches the node $v$, the degree of the node $v$.  If two robots decide to move along the same edge from the same or different end nodes in the same round, none of the robots can detect the other robot's movement.

\subsection{Our Contribution}
So far, all the works on \dpn~problem require two types of communication abilities between the robots: (1) when some robots are at the same node, they can communicate by exchanging messages between them (2) any two robots in the network can exchange messages between them.  
In this paper, we show that none of the above communication abilities is required to achieve dispersion.  We propose a deterministic algorithm in our described model, that runs in time $O(k \log L+ k^2 \log \Delta)$, where $k$ and $\Delta$ are the number of robots and the maximum degree of the underlying graph, respectively. The additional memory used by the robots is $O(\log L + \log \Delta)$. We further prove that to achieve dispersion in our model, $\Omega(\log L + \log \Delta)$ additional memory is necessary.

It can be noted here that with the sufficient amount of memory available to each robot, dispersion can be achieved by the following trivial algorithm: the robot with id $i$ starts moving according to depth-first search in round $i$ where the edges are visited in the increasing order of their port numbers. The robot settles down at the first node which is not previously occupied by some other robot.

Hence, our main contribution in this paper is to design an efficient dispersion algorithm by a set of silent robots with asymptotically optimal memory.


\subsection{Related Works} The \dpn~problem was first introduced by Augustine and Moses Jr. \cite{AugustineM18}. In this paper, the authors considered this problem when the number of robots (i.e., $k$) is equal to the number of vertices (i.e., $n$) and all the robots are initially co-located. Along with arbitrary graphs, they also study various special classes of graphs such as paths, rings, and trees.  They proved that, for any graph $G$ of diameter $D$, any deterministic algorithm must take $\Omega(\log n)$ bits of memory by each robot and $\Omega(D)$ number of rounds. For arbitrary graphs with $m$ edges, they provided an algorithm that requires $O(m)$ rounds where each robot requires $O(n\log n)$ bits of memory. For paths, rings, and trees, they provided algorithms such that each robot requires $O(\log n)$ bits of memory and takes $O(n)$ rounds. Further, their algorithm takes $O(D^2)$ rounds for  rooted trees, and each robot requires $O(\Delta + \log n)$ bits of memory. All their algorithms work under the local communication model, i.e., co-located robots can communicate among themselves.

Kshemkalyani and Ali \cite{KshemkalyaniF19} proposed five different \dpn~algorithms for general graphs starting from arbitrary initial configurations. Their first three algorithms require $O(m)$ time   and each robot requires $O(k \log \Delta)$ bits of memory, where $m$ is the number of edges and $\Delta$ is the degree of the graph. These three algorithms differ on the system model and what, where, and how the used data structures are maintained. Their fourth and fifth algorithms work in the asynchronous model. Their fourth algorithm uses $O(D \log \Delta)$ bits of memory at each robot and runs in $O(\Delta ^D)$ rounds, where $D$ is the graph diameter. Their fifth algorithm uses $O(\max(\log k, \log \Delta ))$ bits memory at each robot and uses $O((m-n)k)$ rounds. All their algorithms work under the local communication model.

In \cite{KshemkalyaniMS19}, Kshemkalyani et al. provided a novel deterministic algorithm in arbitrary graphs in a synchronous model that requires $O(\min(m, k \Delta) \log k)$ rounds and $O(\log n)$ bits of memory  by each robot. However, they assumed that the robots know the maximum degree and number of edges.  Shintaku et al. \cite{ShintakuSKM20} studied the \dpn~problem of \cite{KshemkalyaniMS19} without the knowledge of maximum degree and number of edges and provided an algorithm that uses the same number of rounds and $\log(\Delta + k)$ bits of memory per robot which improves upon the memory requirement of \cite{KshemkalyaniMS19}. Recently Kshemkalyani et al.\cite{Kshemkalyaniarxiv} came up with an improved algorithm where it requires $O(\min(m, k \Delta))$ rounds but the memory requirement remains the same as in  \cite{ShintakuSKM20}. All the algorithms in \cite{KshemkalyaniMS19, ShintakuSKM20, KshemkalyaniMS19} works under the local communication model.

The \dpn~problem was studied on dynamic rings by Agarwalla et al. \cite{AgarwallaAMKS18}. In \cite{MollaMM20}, Molla et al. introduced fault-tolerant in \dpn~problem in a ring in the presence of the Byzantine robots. The results are further extended by the authors in \cite{Anisur21} where dispersion on general graphs in presence of Byzantine robots are considered. In all these algorithms, local communication model is considered. Molla et al. \cite{MollaM19} used randomness in \dpn~problem. They gave an algorithm where each robot uses $O(\log \Delta)$ bits of memory. They also provided a matching lower bound of $\Omega(\log \Delta)$ bits for any randomized algorithm to solve the \dpn~problem. They extended the problem to a general $k$-dispersion problem where $k > n$ robots need to disperse over $n$ nodes such that at most $\frac{k}{n}$ robots are at each node in the final configuration. Very recently, Das et al. \cite{Das2020} studied dispersion on anonymous robots and provided a randomized dispersion algorithm where each robot uses $O(\log \Delta)$ bits of memory. In both the works, local communication model is considered. There are works \cite{KshemkalyaniMS20walcom, Kshemkalyanijpdc22} in the global communication model as well where robots can communicate even if they are located in different nodes. Results in these paper includes dispersion on grids as well as general graphs.

Note that, in all of the results mentioned above, robots need to communicate between them, either locally or globally.


\section{Dispersion on Graphs}
In this section, we propose an algorithm that achieves dispersion in any anonymous graph in time $O(k \log L+k^2\log \Delta)$ and with $O(\log L+\log \Delta)$ {additional} memory. Before we describe our algorithm, we give an overview how previous results on dispersion work where co-located robots with limited memory can exchange arbitrary amounts of messages between them.
The proposed algorithms in the previous works on dispersion rely on exploring nodes of the graph using depth-first search (DFS).
At the beginning, all the robots are at a node $s$ and the smallest id robot settles at $s$ and the other robots move to an adjacent node according to DFS. The robots learn about the smallest id by exchanging their ids among co-located robots.
In any round, the robot with the smallest id among the co-located robots settles at the current node and other robots move to an adjacent empty neighbor. Here, the self-placement of a robot at an `empty' node represents `coloring' of already visited nodes in DFS traversal. Therefore, even if the graph is anonymous, DFS traversal can still be executed by the mobile robots. The difficulty here is that due to limited memory, the robots may not store the entire path it follows while traversing nodes before it settles down at an empty node. To be specific, suppose that, all the unsettle robots reach a node $v$ (which is already occupied by some other robot) during DFS. If the robots observe that each of the neighbor of $v$ is already visited (the robot can learn this by visiting the neighbors of $v$ and observing that these neighbors are already occupied by some other mobile robots), then as per DFS, all the unsettle robots  must backtrack to the node $u$ from which it visited $v$ and then search for another empty neighbor of $u$. If there is no empty neighbor of $u$, backtrack again and continue this way until  an empty node is found. Without sufficient memory, the robots cannot do this process of backtracking by themselves. Here, the capability of communications between co-located robots again comes for rescue. Each robot stores the incoming port through which it enters the empty node where it  settled down. The information of this port serves as the pointer to backtrack from a particular node. When the set of unsettled robots unable to find any empty neighbor of a node $v$, they  have to backtrack. The robot which settled at $v$ provides the information of the port for backtracking to these unsettled robots.
 Therefore, the above DFS like dispersion strategy can be executed with $O(\log \Delta)$ memory at each robot, where $\Delta$ is the maximum degree of a node in the graph.

The difficulty arises when the co-located robots do not have the capability of communication. 
Specifically, the following major issues may arise in the absence of communication.
\begin{itemize}
\item Each mobile robot only knows its own label but unable to know the labels of other robots without direct communication. Therefore, a strategy like a robot with the `minimum' or `maximum' label settled down will not work.
\item With limited memory and lack of communication, a robot may not learn sufficiently long path information which is needed for backtracking.
\end{itemize}

We next describe how our algorithm overcomes the above difficulties and enables the robots to execute the dispersion successfully. Our algorithm runs in several iterations. We call a node $v$ {\it full} in an iteration $j$, if $v$ is occupied by a robot at the end of the iteration $j$. Otherwise, it is called {\it empty}. Also, for any node $v$, the node adjacent to $v$ and connected through the port $i$ from $v$ is denoted by $v(i)$.

In each iteration of our algorithm, except the last iteration, an empty node becomes full and no full node becomes empty. Therefore, if $k$ robots are present at the start node initially, the task of dispersion is completed within $k-1$ iterations. An additional iteration is required to identify the fact that the dispersion is completed.

Each iteration of the algorithm has two phases. In Phase 1, a leader election algorithm is executed and a robot $M$ is elected as the leader from a set of robots $R$. In Phase 2, an empty node is occupied by a mobile robot that is either the robot elected recently or a robot elected previously.

A detailed description of the algorithm with the high-level idea is described below. During the description of each major step of the algorithm, we explicitly mention the purpose of the respective steps, the difficulty of implementing the steps with existing techniques, and how to overcome such difficulties.

\subsection{The Algorithm}
\subsubsection{High level idea and preliminaries} The proposed algorithm executes in several iterations. Each iteration of the algorithm consists of two phases: Phase 1 and Phase 2. In Phase 1, a leader among the robots situated at $s$ is elected. In Phase 2, one empty node is occupied by a robot. Phase 2 requires several communications between robots. Since there is no means of direct communication, we adopt the idea proposed in \cite{BouchardDP20} which enables the robots to communicate between them by utilizing the robot's movement as the tool of communication. We describe later how this process of communication is executed while describing Phase 2 of our algorithm.

Following tasks are collectively executed by the robots in each iteration.
\begin{itemize}
    \item  Execution of  Phase 1.
    \item Execution of Phase 2. This phase includes five major steps.
    \begin{enumerate}
         \item Informing all robots (those who will participate in Phase 2) that Phase 1 ended.
        \item Finding an empty node.
        \item Propagate information whether an empty node was found or not to the robots `participating' in Phase 2.
        \item Movement of the robots for occupying the empty node.
        \item Termination detection.
    \end{enumerate}
\end{itemize}

Each of the above tasks involves the movement of the robots. A robot may decide to move in a particular round to achieve one of the following goals.
\begin{itemize}
    \item To elect the leader.
    \item To transmit some information.
    \item To search or occupy an empty node.
\end{itemize}

The movement of a robot for a specific purpose may create confusion for other robots, who are affected by this movement. This is because a robot may decide to move for finding an empty node but some other robots may learn this movement as the initiation of some message transmission. Therefore special care must be taken to avoid such ambiguity. Our algorithm overcomes such ambiguity by allotting a unique `slot' to each of the five above mentioned steps for Phase 2 and one slot for Phase 1.
 To be specific, each  round in a consecutive block of six rounds is dedicated to exactly one of the above six steps required to execute Phase 1 and Phase 2.
 That is, a robot moving in round $6i+j$ means different to the algorithm than a robot moving in round $6i+j'$ for $j \ne j'$ and $0 \le j,j' \le 5$. For $0 \le j \le 5$, we call a round {\it $j$-dedicated}, if the round number is of the form $6m+j$ for some positive integer $m$ (ref. Fig. \ref{fig:round}).

 \begin{figure}[ht!]
    \centering
\includegraphics[width=.5\textwidth]{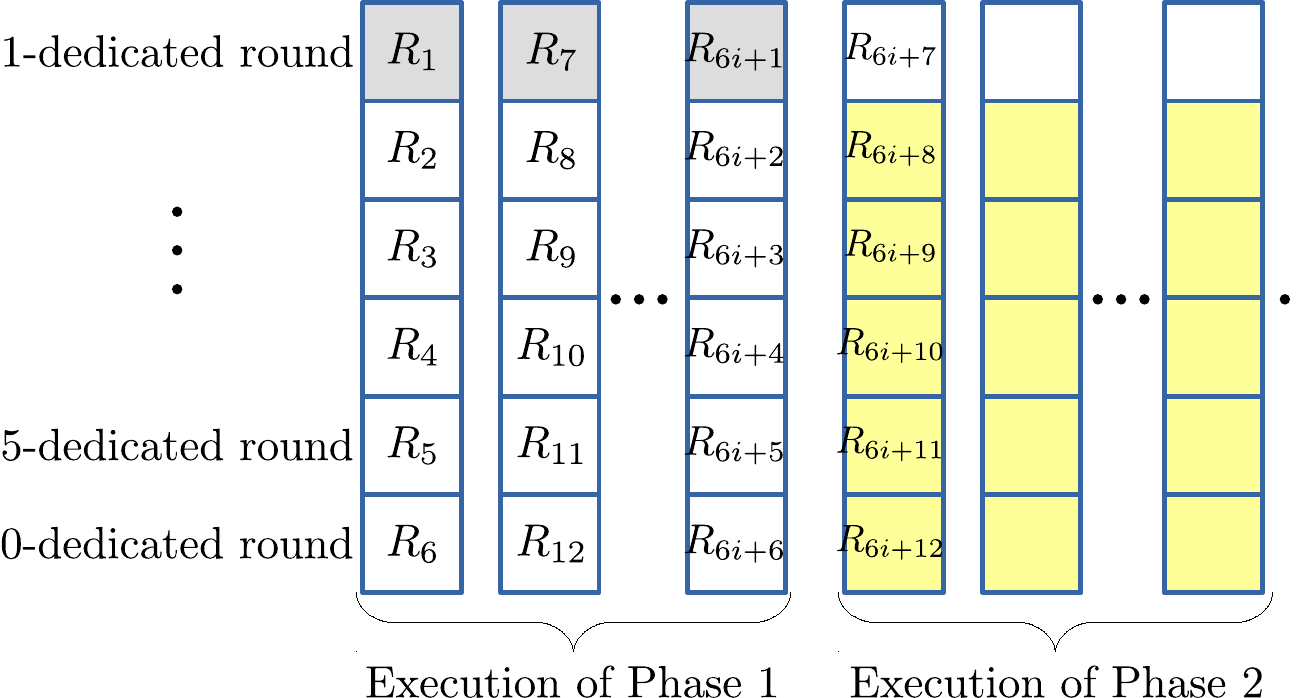}
    \caption{Blocks of consecutive 6 rounds are shown in the figure.
     A round in a block is dedicated to precisely one of the six steps of Phase 1 and Phase 2. The gray colored cells represents 1-dedicated rounds that are used to execute Phase 1  and  yellow colored cells represent eligible dedicated rounds of Phase 2.}
    \label{fig:round}
\end{figure}

During the execution of the algorithm, each robot maintains a variable $status$. At any point of time the $status$  of a robot can be one of following.
\begin{itemize}
\item {$active$:} These robots participates in Phase 1 and one of the $active$ robots is elected as   `leader' in Phase 1 of any iteration.
\item {$master$ and $follower$:} The  robot, elected as the leader in Phase 1 of an iteration, changes its $status$  to $master$ in the very first iteration of the algorithm. In the subsequent iterations, the leader changes its $status$  to $follower$.

\item{$idle$:} A robot with this $status$  does not take part in the algorithm anymore. During the execution of our algorithm, eventually, each robot becomes $idle$.


\end{itemize}
Initially, all robots are $active$.  The algorithm terminates when $status$  of each robot becomes $idle$.
\vspace{0.3cm}

With the above details, we are ready to describe the details of the dispersion algorithm. We start by describing Phase 1 of the algorithm.

\subsubsection{ Phase 1 (Electing Leader)}
Only the $active$ robots start executing this phase. The robots execute the steps of this phase only in 1-dedicated rounds.  The $active$ robots at $s$ start this phase in round 1  if this is the first iteration of the algorithm. Otherwise, if $decrease=true$ in some 5-dedicated round at $s$ (this event signifies the fact that the activities in the last iteration are ended), then all the $active$ nodes start executing Phase 1 in the next 1-dedicated round.

\begin{algorithm}[ht!]
  \SetAlgoLined
    \If{$status=active$}
        {\textsc{Phase\_1(M)}}
    {\textsc{Phase\_2(M)}}\\
    \If{$status \ne idle$}
    {\textsc{Dispersion(M)}}
  \caption{\textsc{Dispersion(M)}}
  \label{alg0:dispersion}
\end{algorithm}

\begin{algorithm}[ht!]
  \SetAlgoLined
    \If{($round=1$) or ($decrease=true$ in a 5-dedicated round) }{
    {Wait for the first available 1-dedicated round.}\\
    {$j=1$, $election=false$.}\\
    \While{$election=false$}{
        {\pb $(M,j)$}\\
        {j=j+1}
    }
    }
  \caption{\textsc{Phase\_1(M)}}
  \label{alg:phase1}
\end{algorithm}

\begin{algorithm}[ht!]
  \SetAlgoLined
    \If{$status=master$}{
    {\master $(M)$.}\\
    }
    \Else{
        \If{$status=follower$}
            {\follower $(M)$.}
        \Else{
        \If{$status=active$}
            {\actv $(M)$.}
        }
    }
  \caption{\textsc{Phase\_2(M)}}
  \label{alg:passive}
\end{algorithm}

\begin{algorithm}[ht!]
  \SetAlgoLined
    \tcp{\bf 1st 1-dedicated round:}
    \If{$engage=true$}{
        \If{$alone=true$}{
            {$candidate=true$}}
    \Else{
        \If{ the $j$-th bit of $l(M)$ is 1}
            {$move=1$. Move through port 0. Let $q$ be the incoming port at $s(0)$.}
    }
    }
    \tcp{\bf 2nd 1-dedicated round:}
        \If{$engage=true$}{
        \If{($move=0$) and ($decrease=true$ in the last 1-dedicated round)}
            {$move=2$. Move through port 0.}}
    \tcp{\bf 3rd 1-dedicated round:}
           \If{$engage=true$}{
            \If{$move=2$}{move through port $q$}
            \If{($move=1$) and ($increase=true$ in the 2nd 1-dedicated round)}
                {Move through port $q$.}
            \Else{
                \If{$move=1$ and $increase=false$}
                    {$move=0$. Move through port $q$.}
           }
        }
        \tcp{\bf 4th 1-dedicated round:}
                \If{$engage=true$}{
                \If{$move=1$}
                    {Set $engage=false$. Move through port $\delta-1$. Let the incoming port is $q'$.}}
     \tcp{\bf 5th 1-dedicated round:}
     \If{$candidate=true$}
        {$election=true$. Move through port $\delta-1$. Let the incoming port be $q'$.}

    \tcp{\bf 6th 1-dedicated round:}
        \If{$candidate=true$}{
        \If{ the current iteration is the 1st iteration}
            { Move through port $q'$. $status=master$, $recent=true$.}
        \Else{$status=follower$, $recent=true$. Move through port $q'$.}
        }
        \Else{
            \If{($engage=false$ and $increase=true$ in the  5th 1-dedicated round) }
               {$engage=true$. $election=true$, $move=0$. Move through port $q'$.}
        }
  \caption{\pb $(M,j)$}
  \label{algo:processbit}
\end{algorithm}

Without the ability of communication by message passing, the labels of the robots are used to elect a leader among co-located robots at $s$. On a high level, a robot `leave' $s$ if a bit of its label is 1, else stay.
 Since the labels of two robots are different, there must be at least a position in their labels where the bits are different. Therefore, if, an robot moves when a bit of its label is 1 and does not move if the bit is 0 or all the bits of its labels are already processed, then the two robots must occupy different nodes within $z$ rounds where $z$ is the length of the label of the node with maximum id.  Careful implementation of this process ensures that after $z$ rounds, exactly one robot stays at $s$ and elects itself as the leader.
 However, this raises a difficulty. The difficulty is that the labels of the robots are not necessarily of  the same length. For example, suppose that the label of one robot is `10' and the other is `100'. If the robots process the bits of their label from left to right, then in this case, there is no way one can identify the position of the labels where the bits differ. This difficulty can be overcome by processing the bits from right to left instead of left to right as no label can end with a zero if read from right to left.

Let $l(M)$ be the reverse binary string corresponding to the label of the robot $M$.  Also, let $\delta$ be the degree of $s$. Without loss of generality, we assume that the degree of $s$ is at least 2. Otherwise, each robot moves to the adjacent node of $s$ and starts the algorithm from the node.

Each robot $M$ executes a subroutine called \pb $(M,j)$ (ref. Algorithm \ref{algo:processbit})  for the $j$-th bit of its reverse label $l(M)$, for all $j=1,2,3 ,\cdots$ until the leader is elected. If $|l(M)| <j$ for a robot $M$, then the $j$-th bit is treated as zero while executing \pb $(M,j)$.

Each call of subroutine \pb~is executed for six consecutive 1-dedicated rounds. On a high level, after executing subroutine \pb $(M,j)$, two robots whose $j$-th bits are not the same, gets separated. To be specific, after executing subroutine \pb $(M,j)$, one of the following events happens.
\begin{itemize}
    \item If all robots at $s$ have the same $j$-th bit, then after executing subroutine \pb $(M,j)$ all of them stay at $s$ and participate in the next call of the subroutine for the $(j+1)$-th bit.
    \item Otherwise, the robots whose $j$-th bit are 1 move to $s(\delta-1)$ and remain there until a robot is elected as the leader. Other robots with $j$-th bit 0, stay at $s$ and participate in the next call of this subroutine for the $(j+1)$-th bit.
\end{itemize}

Executing the above steps for $j=1,2, \ldots$, eventually, exactly one robot remains at $s$. This robot changes its $status$  to $master$ in the first iteration and  $follower$ in subsequent iterations.
After learning that the leader is elected, the other robots return to $s$, and Phase 1 of the current iteration ends with this. The details of the subroutine  \pb~are described below.

\noindent {\bf Description of the subroutine \pb}: The robots at $s$ use the set of variables $engage$, $move$, $candidate$, and $election$. At the beginning of Phase 1 of any iteration, for all the robots at $s$, $engage=true$, $move=0$, $candidate=false$ and $election=false$. At the beginning of each call of the subroutine \pb, all the robots at $s$ have $engage=true$ and these robots only participate in the first four 1-dedicated rounds. In the first 1-dedicated round, if the number of robots at $s$ is more than one, then the robots with the $j$-th bit of its reverse label 1 move through port 0 by setting the variable $move=1$. This activity in  round 1 resulted in a possible `split' in the set of robots at $s$. The robots with $j$-th bit 0 stayed at $s$ and the other robots move to $s(0)$. Hence, the robots who stayed at $s$ observe $decrease=true$ in the next round learns about this split. However, the robots that moved to $s(0)$ do not have any idea about this split as it may happen that all the robots moved to $s(0)$.

If some of the robots at $s$ have $j$-th bit 0 and others have $j$-th bit 1, a split happens in the 1st 1-dedicated round. The 2nd 1-dedicated round is for the robots at $s(0)$ to learn about this split. For this purpose, the robots at $s$, move to $s(0)$ in the 2nd 1-dedicated round after setting $move=2$. The robots at $s(0)$ (came to $s(0)$ in the first 1-dedicated round) observe $increase=true$ and hence learn the fact that some robots from $s$ visited $s(0)$ and therefore a split happened in the 1st 1-dedicated round.
The 3rd 1-dedicated round is for all the robots currently at $s(0)$ to come back to $s$. In the 4th 1-dedicated round, the robots whose $j$-th bit are 1, move to $s(\delta-1)$ and set $engage=false$.

If all robots at $s$ have the $j$-th bit 0, then in the 1st 1-dedicated round no robots move from $s$ and each of them has $move=0$.  Since these robots observe $decrease=false$, they learned that no split happened and did not move in the 2nd 1-dedicated round, and hence no robot participates in the 3rd or 4th 1-dedicated round. Hence at the end of the 4th 1 dedicated round, all the robots at $s$ have $move=0$, $engage=true$.

If all the robots at $s$ have the $j$-th bit 1, then in the 1st 1-dedicated round all robots move from $s$ to $s(0)$ and each of them has $move=1$. In the 2nd 1-dedicated round, no robots visit $s(0)$ (as there is no robot has left $s$). Hence the robots at $s(0)$ see $increase=false$ after the 2nd 1-dedicated round and learn that no split happened at $s$. In the 3rd 1-dedicated round these robots move back to $s$ and set $move=0$. Hence no robot participates in the 4th 1-dedicated round.

The 5th and 6th 1-dedicated rounds are participated by the robots only if there was only one robot at $s$ in the 1st 1-dedicated round. In this case, this robot had set $candidate=true$, and the subroutine identifies this robot as the leader. At this point, all the other robots must have set $engage=false$ and are in $s(\delta-1)$. To `inform' the robots at $s(\delta-1)$ that the leader is elected, the robot at $s$ (with $candidate=true$ ) moves to $s(\delta-1)$. Hence, in the 6th 1-dedicated round, robots at $s(\delta-1)$, after observing $increase=true$ in the 5th 1-dedicated round, learn that the leader is elected, and move back to $s$ after setting $election=true$, $engage=true$ and $move=0$. The robot with $candidate=true$ also returns to $s$ in the 6th 1-dedicated round and changes its $status$  to $master$ if the current iteration is the first iteration of the algorithm, else changes its $status$  to $follower$.

\subsubsection{ Phase 2 (Occupying an empty node)} All the robots except those who became $idle$ participate in this phase. In this phase, one empty node is occupied by a robot. On a high level, let $v_1, v_2, \cdots, v_{p-1}$ be the nodes that became full in consecutive iterations and $v_1$ is a neighbor of $s$. Let $r_1, r_2, \cdots, r_{p-1}$ be the robots that are in $v_1, v_2, \cdots, v_{p-1}$, respectively, before Phase 2 of the current iteration starts. Then the leader elected in Phase 1 of the current iteration moves to $v_1$, $r_1$ moves to $v_2$, $r_2$ moves to $v_3$ and so on and, finally $r_{p-1}$ moves to an empty neighbor of $v_{p-1}$.


During the execution of our algorithm, there is a unique $master$ robot, and the other robots are either $follower$, or $active$, or $idle$.

\begin{algorithm}[ht!]
  \SetAlgoLined
    \If{$recent=true$}
        {
        {$prt=child+1$. Go to step \ref{master:step1}}
        }
    \Else{
    {Wait until $increase=true$ in a 0-dedicated round or in a 2-dedicated round}\\
    \If{$increase=true$ for a 0-dedicated round}
    {
    {$status$ =$idle$}
    }
    \Else{
    {$prt=0$}\\
    {  $Found=false$}\label{master:step1}\\
    \While{$Found=false$}{
    { Move through port $prt$ in the next available 3-dedicated round. Let the incoming port is $q$}\\
    \If{$alone=true$}
    {
    {$Found=true$}
    }
    \Else{
    {Move through port $q$ in the next round. $prt=prt+1$}\\
    }
    }
    \If{$Found=false$}{
    {
    $\alpha=1111$.}\\
    {\ssl $(\alpha,parent)$}\\
    {$status=idle$}
    }
    \Else{
    {$\alpha=1011 \cdot B_{prt}$, where $B_{prt}$ is the transformed binary encoding of the binary representation of the integer $prt$.}\\
    \If{$parent\ne NULL$}
    {\ssl $(\alpha,parent)$}
    {  $parent=q$, $recent=false$. Move through port $prt$ in the next available 5-dedicated round.}
    }
    }
    }
  \caption{\master ($M$)}
  \label{algo:master}
\end{algorithm}

\begin{algorithm}[ht!]
  \SetAlgoLined
   \If{$forward=0$}{
        \If{$alone=true$}{
            {Move through port port $child$ in the next available 0-dedicated round. Let the incoming port is $q$ }\\
            {Move through port $q$. $status=idle$.}\
    }
   \Else{

     {Move through port $child$ in the next 2-dedicated round. Let the incoming port be $q$. }\\
     {Move through port $q$ in the next round.}\\

   }
   }
    \Else{
        {Wait until $increase=true$ in a 0-dedicated round or in a 2-dedicated round}\\
    \If{$increase=true$ in a 0-dedicated round}
    {{Move through port $child$ in the next 0-dedicated round. Let the incoming port be $q$}\\
    {Move through port $q$ in the next 0-dedicated round. $status$ =$idle$}
    }
    \Else{
    {Move through port $child$ in the next 2-dedicated round. Let the incoming port be $q$. }\\
    {Move through port $q$ in the next round}\\
    }
    }
    [$\gamma,p$]=\lsl $(M)$

    \If{$\gamma=1111$}{
        $status= master$. Call \master ($M$).}
    \Else{
    {\ssl $(1110 \cdot B_{child} ,parent)$}\\
    { Move through the port $p$ in the next available 5-dedicated round. Let the incoming port be $q$.}\\
    {$child=p$,$parent=q$. $forward=1$.}
    }
  \caption{\textsc{Follower($M$)}}
  \label{algo:follower}
\end{algorithm}
 \begin{algorithm}[ht!]
  \SetAlgoLined
    {[$\gamma,p$]=\lsl $(M)$}\\
    {$child=p$.}\\
    \If{$\gamma=1111$}{
        \While{$decrease=false$ in a 5-dedicated round}{
            \If{$decrease=true$ in a 3-dedicated round}
                {$child=child+1$}
        }
    }
\caption{\textsc{Active(M)}}
  \label{alg:passive}
\end{algorithm}

\begin{algorithm}[ht!]
  \SetAlgoLined
   {Wait until $increase=true$ in a 4-dedicated round}\\
    {$signal=0, b_1=0,b_2=0$}\\
    {$\alpha'=\epsilon$}\\
    \While{$signal \ne 1$}{
    \If{$increase=true$ in the last 4-dedicated round}
        {$b_1=1$}
    {Wait for the next 6 rounds}\\
    \If {$increase=true$ in the last 4-dedicated round}
        {$b_2=1$}
    \If{$b_1=0$ and $b_2=0$}
        {$signal$ =1}
        \Else{
        {$\alpha'=\alpha'\cdot b_1b_2$}\\
        {$b_1=b_2=0$}}
        {Wait for the next 6 rounds}
    }
    {Let $\alpha'= \gamma \cdot\beta'$, where $\gamma$ is the first two bits of $\alpha'$. Let $\beta$ be the string obtained from $\beta'$, respectively, by replacing every 11 by a 1 and every 10 by a 0 from left to right. Let $p$  be the integers whose binary representation is $\beta$. if $\beta=\epsilon$, then set $p=-1$}\\
    {Return $[\gamma,p]$}
  \caption{\lsl $(M)$}
  \label{alg0:learnsignal}
\end{algorithm}

\begin{algorithm}[ht!]
  \SetAlgoLined

   {$z=|\alpha|$, $j=1$}\\
   \For{$j=1$ to $z$}
    {
        \If{$j$-th bit of $\alpha$ is 1}{
           { Move through port $p$ in the next available 4-dedicated round. Let the incoming port is $q$.}\\
           {Move through port $q$ in the next round.}\\
   }
   }
  \caption{\ssl $(\alpha,p)$}
  \label{alg:sendsignal}
\end{algorithm}

We define $parent$ and $child$ of the robots with $status$  $master$ and $follower$ in an iteration. Let $master$ (resp.~$follower$) be at node $v$ at the start of Phase 2 of some iteration. The parent of a $master$ robot $r$ (resp.~$follower$ robot) is the node $u$ from where the $master$ (resp.~$follower$) is entered at $v$ for the first time. If the $master$(resp.~$follower$)
is at $s$, then its parent is defined as null. Similarly, the $child$ of a $master$ or $follower$ robot is the node to which the robot will move at the end of the current iteration.

Since the graph is anonymous, the $parent$  and $child$ of a robot are identified by the port numbers through which its parent and child can be reached from
the robot's current position, respectively.
Initially at $s$, all the robots have $parent=NULL$ and $child=0$.


We first describe the high level description of the task the robots collectively execute in Phase 2. Let $M_p$ be the unique $master$ robot at $v_p$. Let $M_{p-1}, M_{p-2}, \cdots, M_1$ be the $follower$ robots at the nodes $v_{p-1}, v_{p-2}, \cdots, v_1$, respectively, such that $v_j$ is the parent of the robot $M_{j+1}$, $v_1=s$ and $M_1$ is the leader that is elected in Phase 1 of current iteration. In this phase, the $master$ robot $M_p$ searches for an empty neighbor of its current node $v_p$. While searching, the $master$ robot visits the neighbors of $v_p$ using the edges incident to $v_p$ in the increasing order of their port numbers, starting from 0 until it finds an empty node. Intuitively, if $master$ finds an empty node, it goes to $v_{p-1}$ and `informs' $M_{p-1}$ to move to $v_p$ and then $M_p$ moves to that empty node. The robot $M_{p-1}$, after learning the information that $M_p$ is going to leave $v_p$, `informs' $M_{p-2}$ to occupy $v_{p-1}$ and then move to $v_{p}$. This procedure of information exchange and moving forward goes on until $M_1$ moves to $v_2$ and with this, Phase 2 ends.

\begin{figure}[ht!]
\begin{center}
\subfigure{\includegraphics[width=.23\textwidth]{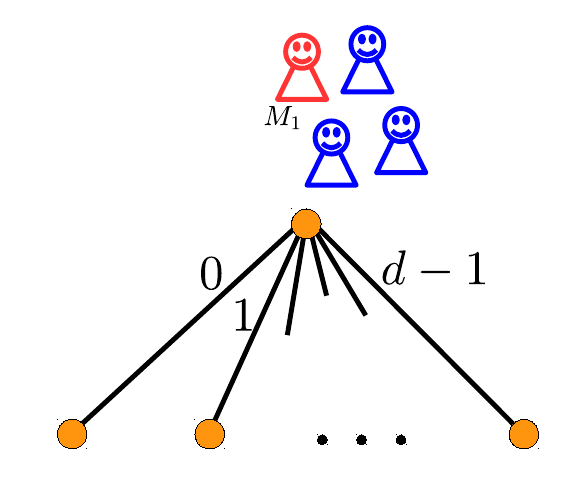}
\label{fig:aaaa-1}
}
\hfill
\subfigure{\includegraphics[width=.23\textwidth]{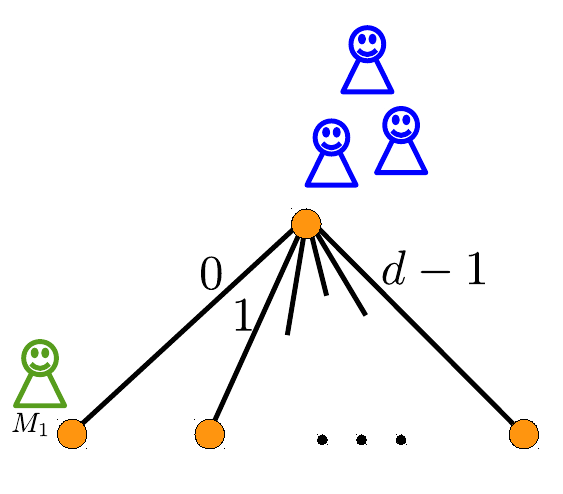}
\label{fig:aaaa-2}
}
\end{center}
\caption{Activities in the first iteration. (a) $M_1$ is elected as the leader   (b) $M_1$ becomes the $master$, finds and occupies the empty node $s(0)$.}
\label{fig:aaaa}
\end{figure}

\begin{figure}[ht!]
\begin{center}
\subfigure[ ]{\includegraphics[width=.23\textwidth]{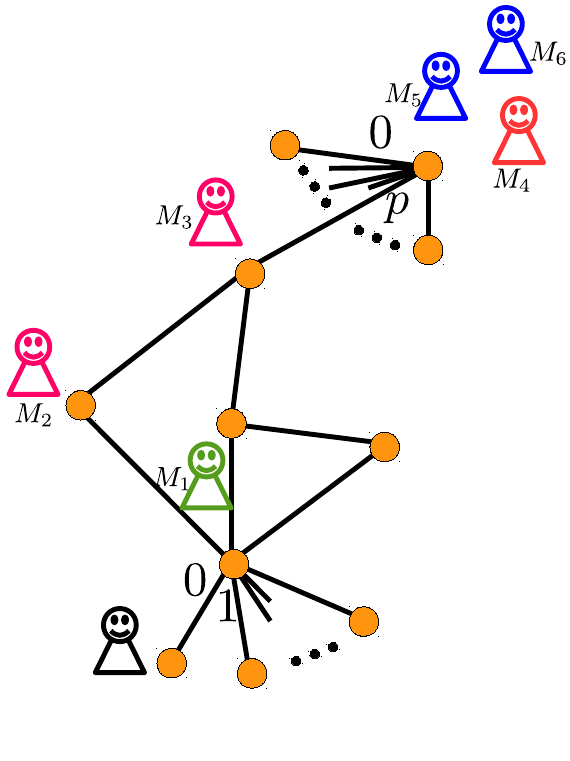}
\label{fig:bbbb-1}
}
\hfill
\subfigure[ ]{\includegraphics[width=.23\textwidth]{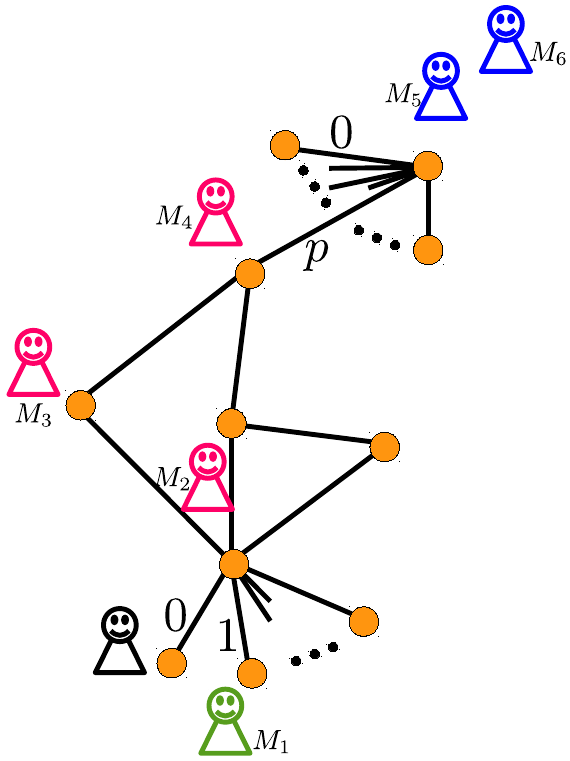}
\label{fig:bbbb-2}
}
\end{center}
\caption{Activities in an intermediate iteration. (a) A leader is elected at the end of Phase 1.  (b) The scenario at the end of Phase 2 when the master robot finds an empty neighbor.}
\label{fig:bbbb}
\end{figure}

\begin{figure}[ht!]
\begin{center}
\subfigure[ ]{\includegraphics[width=.23\textwidth]{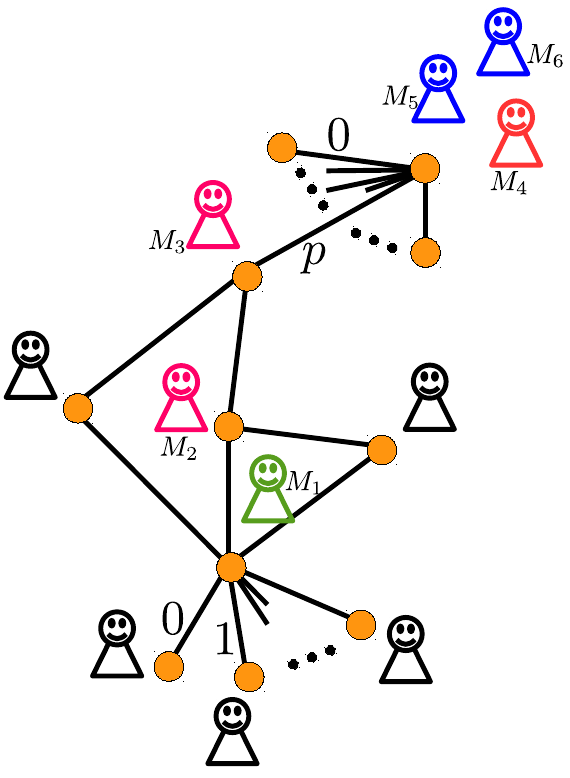}
\label{fig:cccc-1}
}
\hfill
\subfigure[ ]{\includegraphics[width=.23\textwidth]{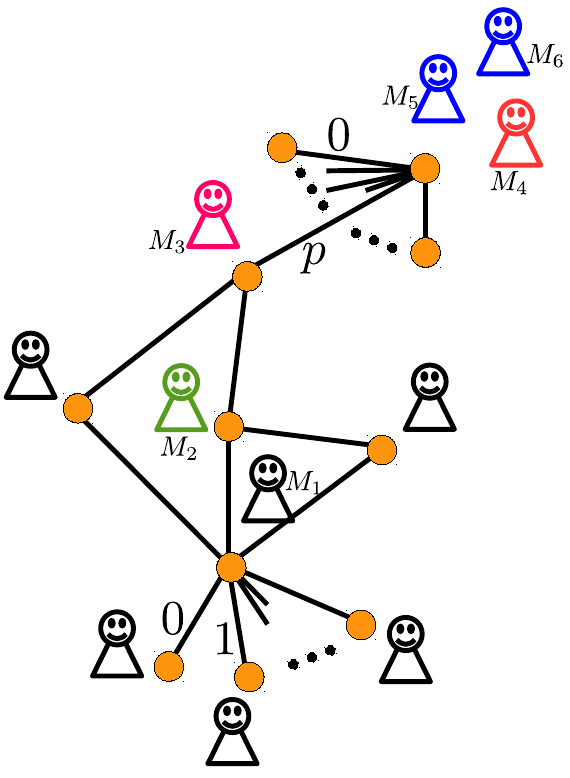}
\label{fig:cccc-2}
}
\end{center}
\caption{Activities in an intermediate iteration when a master robot becomes $idle$. (a) A leader is elected at the end of Phase 1. (b) The scenario when the master does not find any empty neighbor. The master robot becomes $idle$ and the follower robot at the parent of the master robot becomes the new master.}
\label{fig:cccc}
\end{figure}

\begin{figure}[ht!]
\begin{center}
\subfigure[ ]{\includegraphics[width=.23\textwidth]{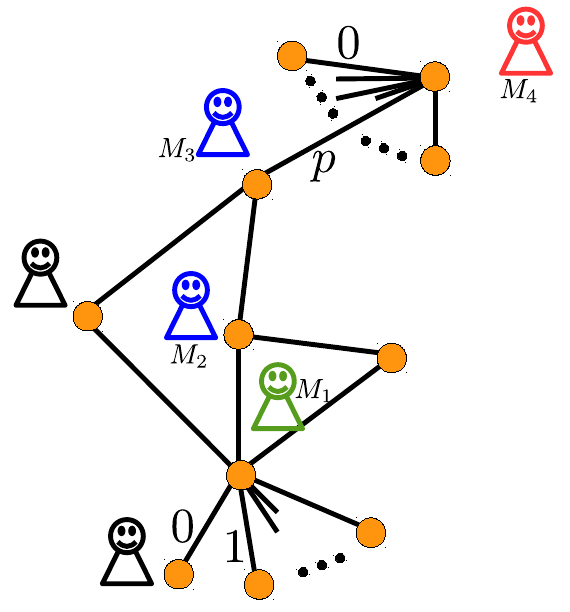}
\label{fig:dddd-1}
}
\hfill
\subfigure[ ]{\includegraphics[width=.23\textwidth]{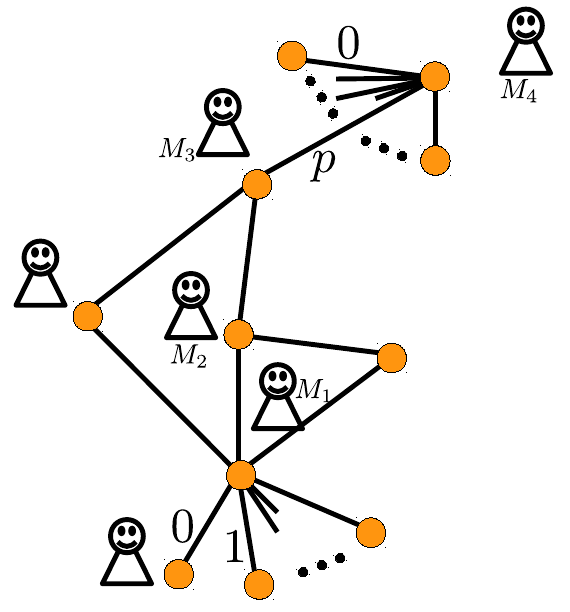}
\label{fig:dddd-2}
}
\end{center}
\caption{Activities in the last iteration. (a) $M_4$ is the only robot present at the root and it becomes the leader.  (b) As a $follower$, $M_4$ sends the termination information that reaches to the $master$ via the all the $follower$s who are present on the path to the $master$. Then $M_4$ occupies the root. Then all robots who were not already $idle$, become $idle$. }
\label{fig:dddd}
\end{figure}

If the $master$ robot does not find any empty neighbor, it informs the same to the robot $r_{p-1}$ and then changes its $status$  to $idle$. The robot $r_{p-1}$, upon learning that the $master$ is  $idle$, changes its $status$  to $master$ and initiates searching for an empty neighbor. This process continues until an empty node is found and then occupied by the $master$ robot.
phase ends when $r_1$ leaves $s$ and occupies a previously occupied node by a $follower$ robot or empty.

There are certain difficulties in implementing the above explained procedures.
\begin{enumerate}
\item How the robot learn through which port from the current node it can reach to its parent?
\item How to propagate `information' to the robot residing at the parent node?
\item How to learn through which port $v_p$ is reachable from $v_{p-1}$?

\end{enumerate}

The learning of the port that leads to the parent by a robot is  gained in the previous iteration itself. To elaborate this, first consider the very first iteration of our algorithm. In Phase 1 of the first iteration, a robot is elected as leader and the $status$  of this robot is $master$. In Phase 2, this $master$ finds $s(0)$ empty and hence moves to $s(0)$. While moving to $s(0)$, this $master$ robot learned the incoming port  $q$ of the edge from $s$ to $s(0)$ and set $parent=q$. Suppose that the $master$ and every $follower$ robot knows their parent port till iteration $t$. Then consider the execution of $t+1$-th iteration. In Phase 1 of this iteration, a robot $r$ is elected as leader and sets its $status$  as $follower$. In Phase 2, the $master$ robot, after finding an empty neighbor, assigns its parent port to the port number through which it entered to its empty neighbor. The $follower$ robots at the time of occupying the new node updates its parent port through which it entered to its new position.

The second and third difficulties can be resolved together as follows. First, we explain below how a robot $r$ can exchange a binary string $\alpha$ with a robot $r'$ which is reachable through the port number $p$ from its current node.

\noindent {\bf Message transmission using the movement of mobile robots:}
Let a robot $r$ decided to transmit the string $\alpha$ to a robot $r'$ reachable from the current node of $r$ through port $p$.

First, $r$ waits for the first available 4-dedicated round. Call this round as the 1st 4-dedicated round.
Then  for each $ i \ge 1$, the robot moves through port $p$ in the $i$-th  4-dedicated round if the $i$-th bit of $\alpha$ is 1 and comes back to $v$ in the next round. On the receiving end of this communication, the robot $r'$  decodes $\alpha$ by identifying the event $increase=true$ in every 4-dedicated round at its current node as a 1  and identifying an event $increase=false$ in every 4-dedicated round as 0.  The difficulty here is how the receiving robot  knows when this process of  communication ends. To overcome this, we use the idea of { \it transformed binary encoding}. For any binary string $\alpha$, replace every `1' by `11' and every `0' by `10'. Note that the transformed binary encoding of any binary string can not contain the sub string `00'. Hence, the robot recognizesthe  observation of two consecutive zeros as the end of transmission.

Using this technique, the $master$ robot transmits one of the the following information to its parent in Phase 2.

\begin{enumerate}
\item I am a $master$ robot, I found an empty neighbor though port $p$. This message is encoded as $1011 \cdot B_{p}$, where the transformed binary encoding of the integer $p$ is denoted by $B_p$.
\item I am a $master$ robot, I did not find any empty neighbor.  This message is encoded as $1111$.
\end{enumerate}

Upon receiving the message (1), the $follower$ robot decodes the integer $p$  transmits the following message to the robot connected through its  parent port:(3)``I am a $follower$ robot, I am going to move forward through port $child$. This message can be encoded as $1011\cdot B_{child}$. The $follower$ then moves through the port $p$, updates $parent$ as the incoming port at the destination node and $child=p$.
If a $follower$ robot receives the message (2), it changes its $status$  to $master$ and start vising each of its neighbors in the increasing order of the port number starting from port $child+1$.  When the $follower$ robot at $s$ receives the message (1) or (3), it moves through port $child$ and with this Phase 2 and hence the current iteration ends. If the $follower$ robot at $s$ receives the message (2), then it changes its $status$  to $master$, and start searching for an empty neighbor starting from port $child+1$. The iteration ends once the robot leaves $s$.


We now give the detailed descriptions of the algorithms of the robots with different $status$  in Phase 2.

\noindent {\bf Description of subroutine {\textsc{Follower($M$)}}:} A robot $M$ with $status$  $follower$ executes the subroutine {\textsc{Follower($M$)}} (Algorithm \ref{algo:follower}). The robot uses a binary variable  $forward$. If the robot is at $s$, i.e., this robot is elected as leader in Phase 1 of the current iteration, then $forward=0$. For the other $follower$ robots, which are not in $s$, have $forward=1$. If the robot at $s$ is not the only robot at $s$, i.e., $alone=false$, then this robot moves through its $child$ port in the next available 2-dedicated round and comes back to $s$ in the next round (steps 5-8). The $follower$ robots which are not is $s$, waits until $increase=true$ in a 0-dedicated round or in a 2-dedicated round. If the robot finds $increase=true$ in a 2-dedicated round, then it learns that the current iteration is not the last iteration and send the same information to the robot present in the adjacent node by moving through its $child$ port in the next 2-dedicated round and comes back to its position in the next round. Once this step is executed, the robot calls subroutine \lsl, where it waits until $increase=true$ in a 4-dedicated round. After that, it learns the message from the robot connected through its child port by identifying $increase=true$ as a 1 in a 4-dedicated round and $increase=false$ as a 0 in the same round until two consecutive 4-dedicated rounds have $increase=false$. It then decodes the integer $p$  and the three  bit string $\gamma$. If $\gamma=1111$, then the robot learns that the robot connected through its child port is a $master$ robot and it is now $idle$. After learning this information, the $follower$ robot changes its $status$  to $master$ and set $recent=true$. After that, it starts executing the subroutine \master $(M)$ (Algorithm \ref{algo:master}). If $\gamma$ is either 1011 or 1110, then the robot calls subroutine \ssl $(\alpha,parent)$, where $\alpha=1110 \cdot B_{child}$. After this transmission is complete, the robot moves through port $p$ in the next 5-dedicated round and updates $parent$ as the incoming port through which it entered the empty node, and $child=p$. In step 4, if the robot observes $increase=true$ in a 0-dedicated round, it learns that the current iteration is the last iteration and the $follower$ robot connected through its parent port is now $idle$. It moves through port $child$ in the next available 0-dedicated round and changes its $status$  to $idle$.

\noindent {\bf Description of the subroutine \master ($M$):} A robot with $status$  $master$ executes the subroutine \master ($M$) (Algorithm \ref{algo:master}). If $recent=true$, then the robot was a $follower$ robot at the beginning of the current iteration and changed its $status$  to $master$ because the previous $master$ robot did not find any empty neighbor in  the current iteration and is now $idle$. In this case, the robot executes from step 9 of the algorithm, according to which, it starts searching for an empty neighbor in its neighborhood (steps 9-15). If $recent=false$, then the robot waits until $increase=true$ in a 0-dedicated round or in a 2-dedicated round. If it observes $increase=true$ in a 2-dedicated round, then it learns that the current iteration is not the last iteration of the algorithm. It then starts searching for an empty neighbor in its neighborhood (steps 9-15). If an empty neighbor is found through port $p'$, then the robot executes the subroutine \ssl~with $\alpha=101 \cdot B_{p'}$ through the port $parent$. It then moves to its empty neighbor in the next 5-dedicated round through the port $p'$, updates $child=0$, $parent$ as the port through which it entered to this empty node, and sets $recent=false$.

\noindent{\bf Description of subroutine \actv ($M$):} The $active$ robots at $s$ execute the subroutine \actv $(M)$ in Phase 2. The robots will execute the subroutine \lsl~to learn the  port $p$ and $\gamma$ which was transmitted by the robot connected through its child port by executing the subroutine \ssl. If $\gamma=1111$, then the robots at $s$ learn that the robot connected through its child port was a $master$ robot, and it does not find any empty neighbor. Therefore, the $follower$ robot at $s$ is going to become $master$ and will start searching for an empty neighbor starting from port $child$. Since 3-dedicated rounds are used in search of an empty neighbor by a $master$ robot, the $active$ robots observe how many times the event $decrease=true$ in a 3-dedicated round occurs at $s$. If there are total $j$ times $decrease=true$ occurs then that signifies $s(child+1), s(child+2), \cdots, s(child+j)$ are full. Accordingly, the robots at $s$ update their $child=child+j+1$.

A flow chart of the different phases in an iteration and different steps inside a phase is shown in figure \ref{fig:flowchart}.

\begin{figure*}[ht!]
    \centering
\includegraphics[width=1\textwidth]{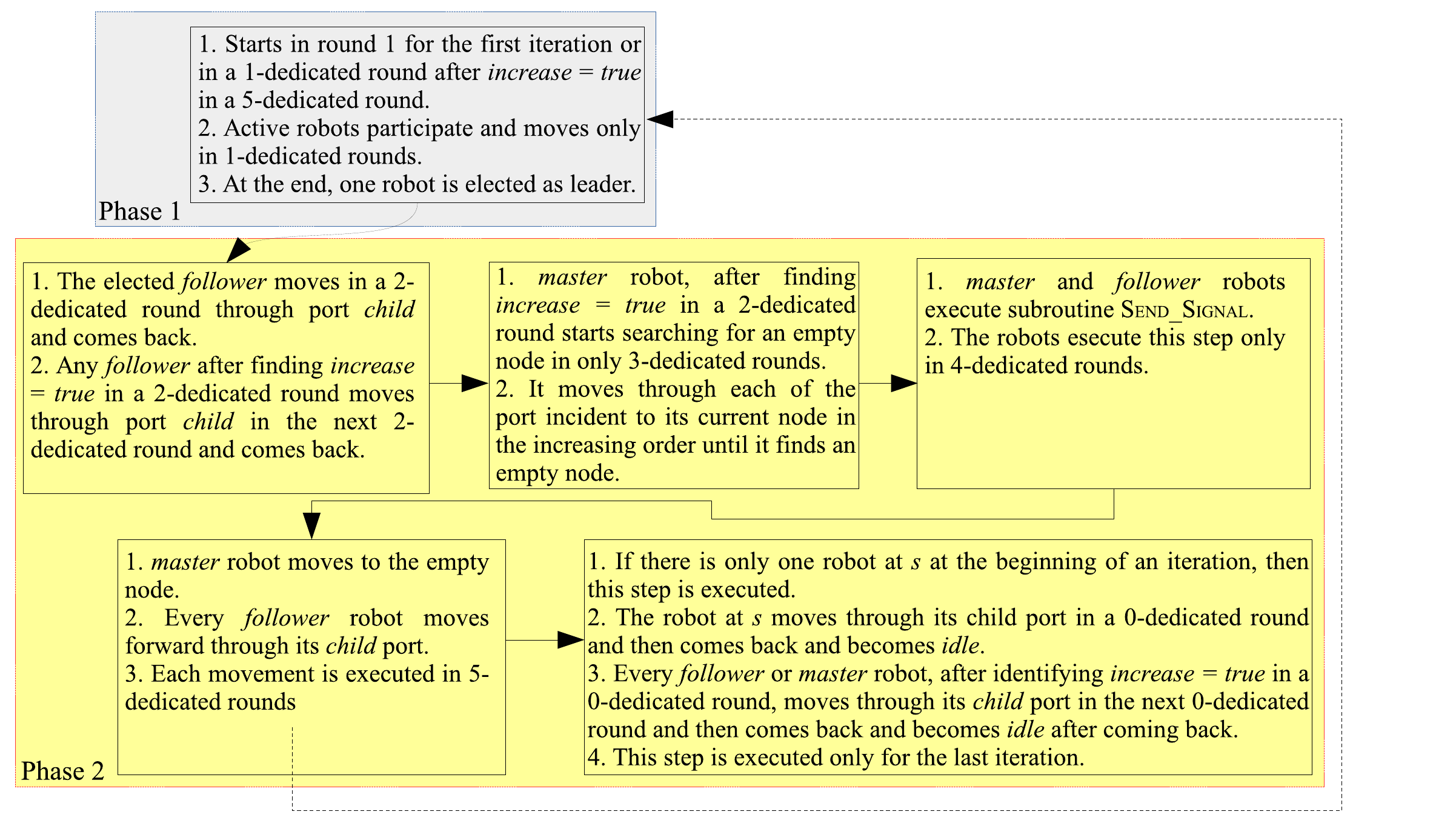}
    \caption{Activities in Phase 1 and Phase 2 of an iteration}
    \label{fig:flowchart}
\end{figure*}

\subsection{Correctness and Analysis}
The following lemma will be useful to show that a unique leader is elected in Phase 1 of any iteration.
\begin{lemma}
In Phase 1 of any iteration of our algorithm, for $j \ge 1$, let $U_j$ be the set of robots at $s$ and $U'_j$ be the set of robots at $s(\delta-1)$ before the call of the subroutine \pb $(M,j)$ by every robot in $U_j \cup U'_j$. If $|U_j|>1$ then, the following statements are true.

\begin{enumerate}\label{lem:phase1}
    \item All the robots in $U_j$ have $engage=true$,  $candidate=false$,  $move=0$, $election=false$.
    \item All the robots in $U'_j$ have $engage=false$,  $candidate=false$,  $move=1$, $election=false$.
    \item Let $ U_j(1)$ be the robots in $U_{j}$ with $j+1$-th bit at its reverse label 1 and $U_j(0)=U_j \setminus U_j(1)$. If $U_j(1)=U_j$ or $U_j(0)=U_j$,
    then, $U_{j+1}=U_{j}$ and $U'_{j+1}=U'_{j}$; otherwise, $U_{j+1}=U_j(0)$, $U'_{j+1}=U'_{j}\cup U_j(1)$.

\end{enumerate}
\end{lemma}

\begin{proof}
We prove the above statements using induction. Consider the execution of the algorithm for an iteration when there is at least two robots in $s$. Let $U$ be the set of all the robots at $s$.
Phase 1 of this iteration is executed by all the robots in $U$ where each robot $M$ in $U$ calls the subroutine \pb $(M,j)$  for $j=1,2,3, \cdots$ until $election=true$.

The statement (1) and (2) are trivially true for $j=1$, where $U_1=U$ and $U_1'=\emptyset$. To prove the base case for (3), consider the execution of the subroutine \pb~for $j=1$.
 Since $engage=true$  and $alone=false$ for each robot in $U_1$, each robot performs step 5 of the subroutine \pb~(Algorithm \ref{algo:processbit}) based on the first bit of $l(M)$. Consider the following two cases.

\begin{itemize}
\item {\bf [The first bit is 0 for all robots at $s$]:} In this case $U_1(1)=\emptyset$ and $U_1(0)=U_1$.
None of the robots move in step 6 and for each of the robots, the variable $move$ remains 0. In step 7, the robots found $decrease=false$ as no robots moved in the earlier 1-dedicated round. Therefore, none of the robots move in the 2nd 1-dedicated round as well (steps 8-9). The 3rd,  4th and  5th 1-dedicated  rounds are not executed by any of these robots as $move=0$ and $candidate=false$ for each of them. Hence, the robots stay at $s$ in these three consecutive 1-dedicated rounds. The robots does not execute  the 6th 1-dedicated round as well since  steps 22-26 require $candidate=true$, step 28 requires $increase=true$ ($increase=false$ for all the robots as there is no movement happened in the last 1-dedicated round). Hence $U_{2}=U_1$, $U_2'=U_1'=\emptyset$. Therefore, the statement (3) of the lemma is true for this case for $j=1$ the statements (1) and (2) are true for $j=2$.
\item {\bf [The first bit is 1 for all robots at $s$]:} In this case, $U_1(1)=U_q$ and $U_1(0)=\emptyset$. All  the robots in $U_1$ move to $s(0)$ in step 6  as the first bit is 1 for all the robots and  for each of the robots, the variable $move$ becomes 1. In the 2nd 1-dedicated round, these robots do nothing as $move=1$ for each of them (step 7).  Therefore, no  robot moves from $s$ to $s(0)$ in the 2nd 1-dedicated round and all the robots have $move=1$ after 2nd 1-dedicated round.
In the 3rd 1-dedicated round, the robots in $s(0)$ (with $move=1$) set $move=0$ and move back to $s$ as they found $increase=false$ in the 2nd 1-dedicated round. Similar to the earlier case, the robots do not participate in the 4th, 5th and the 6th 1-dedicated round of this subroutine. Hence $U_{2}=U_1$, $U_2'=U_1'=\emptyset$. Therefore, Hence the statement (3) of the lemma is true for this case for $j=1$ the statements (1) and (2) are true for $j=2$.

\item {\bf [At least two robots have different first bit]:} In this case both $U_1(1)$ and $U_1(0)$ are non empty.
The set $U_1(0)$ is the set of robots that have the first bit 0 in the reverse binary label and $U_1(1)$ is the set of robots that have the first bit 1 in the reverse binary label. In the 1st 1-dedicated round, all the robots in $U_1(0)$ stay at $s$ and the robots in $U_1(1)$ move to $s(0)$. Thus, the set of robots $U_1(0)$ observe $decrease=true$ in the 1st 1-dedicated round. Hence, in the 2nd 1-dedicated round, the robots at $s$ (the set of robots $U_1(0)$ with $move=0$)   move to $s(0)$ after setting $move=2$. In the 3rd 1-dedicate round, the set of robots $U_1(0)$  with $move=2$ returns back to $s$ and do nothing in the 4th 1-dedicated round. The set of robots $U_1(1)$, after observing  $increase=true$ in the 3rd 1-dedicated round  return back to $s$ and then in the 4th 1-dedicated round, they move to $s(\delta-1)$ after setting $engage=false$. Hence at the end of 4th 1-dedicated round, the robots in $s$ are the set of robots in $U_1(0)$ with $engage=true$,  $candidate=false$,  $move=0$, $election=false$. The robots in $s(\delta-1)$ are the set of robots $U_1(1)$ with $engage=false$,  $candidate=false$,  $move=1$, $election=false$.
 None of the robots in $U_1(0)$ and $U_1(1)$ execute the 5th and 6th 1-dedicated round, since $candidate=0$ for each of them. Therefore, $U_2=U_1(0)$ and $U_2'= U_1(1)=U_1' \cup U_1(1)$. This proves that the statement (3) of the lemma is true for $j=1$ and the statements (1) and (2) are true for $j=2$.
\end{itemize}

Therefore, the statements of the lemma are true for $j=1$.

Suppose that the statements (1)-(4) are true for some integer $j'>1$. Consider the execution of the subroutine \pb~for the integer $j'+1$. By induction hypothesis, before calling subroutine \pb~for $j'+1$, all the robots in $U_{j'}$ have $engage=true$,  $candidate=false$,  $move=0$, $election=false$.


     Since $engage=true$  and $alone=false$ for each robot in $U_j$ (as number of robots in $s$ is more than 1), each of these robots performs step 5 based on the $(j'+1)$-th bit of its reverse binary label.

     Consider the following   cases.
    \begin{itemize}
        \item {\bf [None of the robots in $U_j$ have the $(j'+1)$-th bit 1]:} In this case $U_{j'}(1)=\emptyset$ and $U_{j'}(0)=U_{j'}$. Note that the labels of the robots in $U_{j'}(0)$ are of two types: (a) the labels have $(j'+1)$-th bit 1 in case the length of the label is at least $j+1$ (2) the labels with length at most $j'$.
        None of these robots in $U_{j'}(0)$ moves in step 6 and hence for each of these robots, the variable $move$ remains 0. In step 7, the robots in $U_j$ found $decrease=false$ as no robots moved in the earlier 1-dedicated round from $s$. Therefore, none of the robot moves in the 2nd 1-dedicated round as well (steps 8-9). The 3rd, 4th and  5th 1-dedicated rounds are not executed by any of these robots as $move=0$ and $candidate=false$ for each of them. Hence, the robots stay at $s$ in these three consecutive 1-dedicated rounds. The robots does not execute the 6th 1-dedicated round  since  steps 22-26 require $candidate=true$, step 27 requires $increase=true$ in the last 1-dedicated round ($increase=false$ in the last 1-dedicated round for all the robots as there is no movement happened). Hence $U_{j'+1}=U_{j'}(0)=U_{j'}$, $U_{j'+1}'=U'_{j'}\cup U'_{j'}(1)=U'_{j'}$.
\item {\bf [The $(j'+1)$-th bit is 1 for all robots in $U_j$]:} In this case $U_{j'}(1)=U_{j'}$ and $U_{j'}(0)=\emptyset$.
All  the robots move to $s(0)$ in step 6 and hence for each of the robots, the variable $move$ becomes 1. In step 9, these robots in $U_j$ do nothing as $move=1$ for each of them.  Therefore, no  robot moves in round 2 from $s$ to $s(0)$.  Since in the 2nd 1-dedicated round, the robots $U_j$ in $s(0)$ (with $move=1$) found $increase=false$, therefore, they set $move=0$ and move back to $s$.
Similar to the earlier case, the robots do not participate in the 4th, 5th and the 6th 1-dedicated rounds
of this subroutine. Hence Hence $U_{j'+1}=U_{j'}(0)=U_{j'}$, $U_{j'+1}'=U'_{j'}\cup U'_{j'}(1)=U'_{j'}$.

\item {\bf [At least one robot has $(j+1)$-th bit 1 and $U_{j'}(1) \ne U_{j'}$]:}   In this case $U_{j'}(1)\ne \emptyset$ and $U_{j'}(0) \ne \emptyset$.

In the first 1-dedicated round, all the robots in $U_{j'}(0)$  stay at $s$ and the robots in $U_{j'}(1)$  move to $s(0)$. Therefore, the robots in $U){j'}(0)$ (with $move=0$) observe $decrease=true$ in the 1st 1-dedicated round and learns that a set of robots had left $s$ in the 1st 1-dedicated round. Hence, these robots move to $s(0)$ after setting $move=2$ (according to step 8). In the 3rd 1-dedicated round, the set of robots $U_{j'}(0)$  with $move=2$ returns back to $s$ and do nothing in the 4th 1-dedicated round. The set of robots $U_{j'}(1)$  found $increase=true$ in the 2nd 1-dedicated round  and hence return back to $s$ in the 3rd 1-dedicated round (according to step 12) and then move to $s(\delta-1)$ after setting $engage=false$ in the 4th 1-dedicated round (according to step 19). Hence at the end of the 4th 1-dedicated round, the robots in $s$ are the set of robots in $U_{j'}(0)$ with $engage=true$,  $candidate=false$,  $move=0$, $election=false$ and the robots in $s(\delta-1)$ are the set of robots $U'_{J'} \cup U_{j'}(1)$ with $engage=false$,  $candidate=false$,  $move=1$, $election=false$.
 None of the robots execute the 5th and the 6th 1-dedicated round since $candidate=0$ for each of them. Therefore, $U_{j'+1}=U_{j'}(1)$ and $U_{j'+1}'=U_{j'}\cup U_{j'}(1)$.
\end{itemize}
Therefore by induction, the statements of the lemma are true for all $j$.
\end{proof}

The following lemma proves that one robot is always elected as leader in Phase 1.
\begin{lemma}\label{lem:election}
If $m> 1$ robots are at $s$ in the beginning of some iteration, then at the end of Phase 1 of that iteration, exactly one robot changes its $status$  to either $master$ or $follower$.
\end{lemma}
\begin{proof}
Consider two robots in $U_j$ executing subroutine \pb~for the integer $j$ such that either they have different $j$-th bit in their reverse label or one robot has 1 in the $j$-th bit of its reverse label and the length of the reverse label of the other robot is less than $j$. By statement (3) of Lemma \ref{lem:phase1}, exactly one of those two robots stays at $s$ after the execution of the subroutine.
Let $z$ be the length of the reverse label of the robot with maximum id at $s$ at the beginning of Phase 1. We claim that there exists an integer $t \le z$ such that after the execution of the subroutine \pb~for the integer $t$, exactly one robot remains at $s$. In other words, $U_{t}=1$.

To obtain a contradiction, suppose that $|U_{z+1}|>1$ and let $M_1, M_2$ be two robots at $s$ before the execution of the subroutine \pb~for $z+1$. This implies that both $M_1$ and $M_2$ remains at $s$ after each call of the subroutine \pb~for the integers $j=1,2, \cdots, z$. Let the length of the reverse labels of $M_1, M_2$ be $l_1$, $l_2$ respectively such that both $l_1$, $l_2$ are less than or  equal to $z$. Without loss of generality, let $l_1 \leq l_2$. This implies that for each $j=1,2, \cdots, l_1$, both $M_1$ and $M_2$ have the same $j$-th bit on their reverse label, and for $j=l_1+1, l_1+2,\cdots, l_2$, each of the $j$th bits in the reverse label of $M_2$ is 0. If $l_1=l_2$, then this is a contradiction to the unique labels of the robots. Else if $l_1<l_2$, this is a contradiction as all the last $(l_2-l_1)$ bits in $M_2$'s reverse label can not be 0 since binary representation of any integer label must start with a 1.

Now consider the execution of the subroutine \pb~for  for the integer $t$. According to step 3 of the Subroutine \pb~(Algorithm \ref{algo:processbit}), the robot at $s$ finds $alone=true$ in the 1st 1-dedicated round and hence set $candidate=1$. This robot does nothing in the next three 1-dedicated rounds. In the 5th 1-dedicated round, this robot moves to $s(\delta-1)$ (according to step 24 and step 26 of Algorithm \ref{algo:processbit}). It changes its $status$  to  $master$ in the first iteration and $follower$ for the other iterations, in the 6th 1-dedicated round and moves back to $s$. The robots in $U'_{t}$ comes back to $s$ after setting $election=true$, $move=0$, $engage=true$.
Hence, exactly one robot changes its $status$  to either $master$ or $follower$ in Phase of the algorithm.
\end{proof}

The following two lemmas (Lemma \ref{lem:empty} and Lemma \ref{lem:full}) help us to show that some robots always find an empty neighbor for movement in Phase 2.

\begin{lemma}\label{lem:empty}
Before any iteration of the algorithm, if the number of robots at $s$ is at least 2, then  one of the following statements is true.
\begin{enumerate}
    \item The master robot has an empty neighbor
    \item one of the $follower$ robot has an empty neighbor.
    \item $s$ has an empty neighbor.
\end{enumerate}
\end{lemma}
\begin{proof}
We prove the above statement by contradiction. Suppose that none of the above statements is true. Sine $n \ge k$ and the number of robots at $s$ before the iteration is at least 2, there exists at least one node in $G$ which is empty. Let $v$ be an empty node in $G$. There are two possible cases.
\begin{itemize}
    \item {\bf Case 1 (all neighbors of $v$ are full):} Since neighbor of $v$ can neither be $s$ (as statement 3 is not true by assumption), nor have a $master$ or $follower$ robot(as statement 1 and statement 2 are assumed to be false), therefore each of the neighbor of $v$ has robots which are $idle$. Let $v'$ be a neighbor of $v$ and $M$ be an $idle$ robot at $v'$. According to our algorithm,  a  $master$ robot can change its  $status$  to $idle$ in Phase 2 of some iteration, when in some 3-dedicated round if the $master$ robot does not find any empty neighbors (ref. steps 12-14). But this is a contradiction as $v$ is an empty neighbor of $v'$.
    \item {\bf Case 2 (all neighbors of $v$ are empty):} Let $d \ge 2$ be the smallest integer such that there exists a node $v'$ at distance $d$ from $v$ which is not empty and $v''$ is the node  adjacent to $v$ which is empty and at distance $d-1$ from $v$. Therefore, $v'$ is an empty node which has a full neighbor. This is similar to case 1 and will give a contradiction.

\end{itemize}
\end{proof}


\begin{lemma}\label{lem:full}
In any iteration of the algorithm, following statements are true.
\begin{enumerate}
    \item At the beginning of the iteration, there is a simple path $P$ from $s$ to the node where the $master$ node is present and the internal nodes of this path  contains $follower$ robots. Also, for any node $w$ in this path, the next node $w'$ is connected to $w$ through port $child$ of the robot present at $w$.
     \item All the robots which are not in any node in $P$ are $idle$.
    \item If the number of robots at $s$ is at least 2, then exactly one empty node becomes full and no full nodes become empty after the iteration.

\end{enumerate}
\end{lemma}
\begin{proof}

We prove the above two statements using induction on the iteration number. The statements (1) and (2) are trivially true for iteration 1, as there is no $master$ robot at the beginning of iteration 1 and all the robots are in $s$. By Lemma \ref{lem:election}, at the end of Phase 1 of this iteration, exactly one robot at $s$ is elected as leader and changes its $status$  to $master$. Since in the first iteration,  $s(0)$ is empty and $child=0$ for this $master$ robot, according to step 24 of  Algorithm \ref{algo:master}, the $master$ robot moves through the port $child$, i.e., through port 0 and reached $s(0)$ in the next available 5-dedicated round. Since $decrease=true$ in a 5-dedicated round at $s$, therefore Phase 2 of this iteration ends. Hence, $s(0)$ is the only node that becomes full after the 1st iteration.

Suppose that the above statements are true for the $t$-th iteration for some $t \ge 1$. We show that the above statements are true for the $t+1$-th iteration as well. To prove this, consider the execution of the algorithm in the $t$-th iteration. By Lemma \ref{lem:election}, after Phase 1 of this iteration, exactly one robot is elected as leader and changes its $status$  to $follower$. By statement (1), there is a simple path from $s$ to the node where the $master$ node is present and  the robots at internal nodes of this path are $follower$. Also, by statement (2), all other robots which are not in any node in $P$ are $idle$.

Since the number of robots at $s$ is at least 2, by Lemma \ref{lem:empty}, one of the following cases happens. We show that in in each of the following case, after the $t$-th iteration, i.e., at the beginning of the $(t+1)$-th iteration, exactly one robot becomes full and no full node becomes empty.

\begin{itemize}
    \item {\bf [The $master$ has an empty neighbor]:}
      Let $v$ be the current position of the $master$ and let $s,v_1,\cdots,v_t=v$ be the simple path from $s$ to $v$ where the robots at the internal nodes are with $status$  $follower$. Since $forward=0$ for the $follower$ robot at $s$, it moves through port $child$ and reached to the node $v_1$ in the first available 2-dedicated round after Phase 1 (according to step 6 of Algorithm \ref{algo:follower}). The robot at $v_1$ with $status$  $follower$ and with $forward=1$, found $increase=true$ in a 2-dedicated round. Hence, according to step 7 of Algorithm \ref{algo:follower}, moves through port child in the next 2-dedicated round and reached the node $v_2$. Continuing in this way the $follower$ robot at $v_{t-1}$ visits $v$ in a 2-dedicated round. Since the robot at $v$ is with $status$  $master$, it finds $increase=true$ in a 2-dedicated round. Hence by step 9 of Algorithm \ref{algo:master}, the $master$ robot starts searching for an empty node $w$ adjacent to $v$ in the consecutive 3-dedicated rounds. Since $v$ has an empty neighbor $w$ (connected through port $p$,say), the $master$ robot executes the subroutine \ssl $(\alpha, parent)$, where $\alpha=10 \cdot B_{p}$. The movements of the robot corresponding to the subroutine \ssl $(\alpha, parent)$ are done only in 4-dedicated rounds.
      After the subroutine \ssl $(\alpha, parent)$ is executed, the $master$ robot moves to the empty node $w$ in the next 5-dedicated round. The $follower$ at $v_{t-1}$, after finding $increase=true$ in a 4-dedicated round, learns that the robot at the node connected through its $child$ port started executing  the subroutine \ssl $(\alpha, parent)$. Hence it starts executing the subroutine \lsl. According to the subroutine \lsl, the robot
      observes the event of $increase=true$ or $increase=false$ in the consecutive 4-dedicated rounds and accordingly decodes $\alpha''=\gamma \cdot \alpha'$ by identifying $increase=true$ event as a 1 and $increase=false$ as a 0 (ref. steps 5 - 24 of Algorithm \ref{alg0:learnsignal}), where $\gamma$ is the first four bits of $\alpha''$. Once it identifies two consecutive 4-dedicated rounds of $increase=false$, it learns that the transmission of information by the robot connected through the $child$ port has ended. It then computes
    $\alpha$ from $\alpha'$ by replacing every 11 by a 1 and every 10 by a 0 from left to right taking two bits at a time.

    Since $\gamma=1011$, the $follower$ robot learns that the robot connected through the $child$ port was a $master$ robot and it found an empty neighbor. It then computes the integer $p$ whose binary representation is $\alpha$. The robot then it execute \ssl $(01 \cdot B_{child}, parent)$. After this, it moves through port $child$ and updates $child=q$ and starts executing \ssl $(\alpha, parent)$, where $\alpha=01 \cdot B_{child}$.

    Continuing in this way, once the $follower$ robot at $s$ learns $\alpha$, identify the port $p$, and then move through the port $child$ to reach $v_1$ and updates $child=p$. With this movement, Phase 2 and hence the current iteration ends. Hence, at the end of this iteration, the node $w$ becomes full and no full node becomes empty.

    \item {\bf [A $follower$ has an empty neighbor]:} Let $v$ be the current position of the  $master$ robot and let $s,v_1,\cdots,v_t=v$ be the simple path from $s$ to $v$ where the robots at the internal nodes are with $status$  $follower$. Let $c<t$ be the largest integer such that $v_c$ has an empty neighbor. Since the $master$ robot at $v_t$ has no empty neighbor, according to Algorithm \ref{algo:master}, it found $Found=false$ in step 15. Hence, it executes \ssl $(\alpha, parent)$, where $\alpha=1111$.  The $follower$ at $v_{t-1}$, after finding $increase=true$ in a 4-dedicated round, learns that the robot at the node connected through its child port started executing  the subroutine \ssl $(\alpha, parent)$. Hence it observes the event of $increase=true$ or $increase=false$ in the consecutive 4-dedicated rounds and accordingly decodes $\alpha''$ by identifying $increase=true$ event as a 1 and $increase=false$ as a 0 (ref. steps 12- 23 of Algorithm \ref{algo:follower} ). It finds $\gamma$, the first four bits of $\alpha''$ are 1111, hence learns that the robot connected through its $child$ port is a $master$ robot and does not have an empty neighbor. Hence, according to step 34 of Algorithm \ref{algo:master}, it changes its status to $master$ and starts executing Algorithm \ref{algo:master}. If $c=t-1$, then the $master$ now has an empty neighbor and one empty node becomes full by the previous case where $master$ have an empty neighbor. Otherwise, this robot again executes \ssl $(\alpha, parent)$ where $\alpha=1111$. Similar as argued above, the $follower$ at $v_{t-2}$  changes its $status$  to $master$. This continues until the $follower$ at $v_c$ changes its $status$  to $master$ and since it has an empty neighbor,
    by the previous case, exactly one empty node becomes full after the current iteration.
    \item{\bf [$s$ has an empty neighbor]:} In this case, first the $master$ robot executes \ssl $(alpha,parent)$ and then becomes $idle$. Each of the $follower$ one by one becomes $master$ and then $idle$. Eventually, the $follower$ robot at $s$ becomes $master$. Since $s$ has an empty neighbor $w$, by the case where $master$ have an empty neighbor. it moves to $w$ and Phase 2 as well as the current iteration ends. At the end of this iteration, the empty node $w$ becomes full.
\end{itemize}
\end{proof}
The following theorems shows that the algorithm terminates after the $k$-th iteration.
\begin{theorem}\label{theorem:kiteration}
Each robot becomes $idle$ by the $k$-th iteration.
\end{theorem}
\begin{proof}
It is enough to  show that if a robot is not $idle$ after the $(k-1)$-th iteration, then it becomes $idle$ in the $k$-th iteration. Suppose that $Z$ be the set of robots which are not $idle$ at the end of $(k-1)$-th iteration. By Lemma \ref{lem:full}, all the $z$ robots which are not $idle$ are present on the nodes of a simple path from $s$ to the node where the $master$ node is present. Also, for any two consecutive nodes $v$ and $w$ in this path, $v$ is connected to $w$ through port $child$. Let $s,v_1,v_2, \cdots,v_t$ be the path where the robot at $v_t$ is $master$.
Since number robot at $s$ decreases by 1 in each iteration, after the $(k-1)$-th iteration, there is exactly one robot at $s$. Therefore, according to step 2 of Algorithm \ref{algo:follower}, the robot at $s$ moves through the port $child$ in a 0-dedicated round and then becomes $idle$ after coming back to $s$. Since $v_1$ is reachable from $s$ through port $child$, hence the robot at $v_1$ finds $increase=true$ in a 0-dedicated round. Hence, according to step 11 of Algorithm \ref{algo:follower}, this robot again moves through port $child$, reaches $v_2$ and then becomes $idle$ after coming back to $v_1$. Continuing in this way, the $master$ robot at $s$ observe $increase=true$ in a 0-dedicated round and hence it changes its $status$  to $idle$ according to step 5 of Algorithm \ref{algo:master}. Therefore, all the robots in the simple path from $s$ to the node where the $master$ robot is present become $idle$ in the $k$-th iteration.
\end{proof}

Following lemmas and theorem gives the time and memory complexity of the proposed algorithm.

\begin{lemma}\label{lem:1round}
The algorithm executes at most $k\log L$ 1-dedicated rounds.
\end{lemma}
\begin{proof}
According to our algorithm, 1-dedicated rounds are used to executes the task corresponding to Phase 1 of an iteration. In Phase 1 of any iteration, all the robots in $s$ calls the subroutine \pb~for the integers $j=1,2, \cdots$ until $election=true$. Each call of the subroutine \pb~uses six 1-dedicated rounds.
Now, as described in the proof of Lemma \ref{lem:election}, exactly one robot is elected among the robots in $s$ using at most $z+1$ calls of subroutine \pb, where $z$ is the length of the label of the robot with maximum id which is present at $s$ at the beginning of the iteration. Since the ids of the robots in $s$ are in the range $[0,L]$, therefore $z \in O(\log L)$. This proves that Phase 1 of each iteration uses $O( \log L)$ of 1-dedicated rounds. Since exactly one empty node becomes full in each iteration (Lemma \ref{lem:full}), there are at most $k$ iterations in the Algorithm \ref{alg0:dispersion}. Therefore, the total number of of 1-dedicated rounds used in the algorithm is $O(k\log L)$.
\end{proof}

\begin{lemma}\label{lem:2round}
The total number of 2-dedicated round used is $O(k^2)$.

\end{lemma}
\begin{proof}
We first show that, for $1 \le j \le k-1$, the $j$-th iteration uses $O(j)$ 2-dedicated rounds.
In any iteration, the 2-dedicated rounds are used by the $follower$ robots to indicate the fact that Phase 1 of the current iteration ended. Let $P=s,v_1,v_2, \cdots,v_t$ be the simple  path (according to Lemma \ref{lem:full}) where the robot at $v_t$ is $master$ and in each of the internal node of the path, $follower$ robots are present.
The $follower$ robot that is elected in Phase 1 moves through port $child$ in the next available 2-dedicated round from $s$ and then comes back to $s$ (according to step 6 of Algorithm \ref{algo:follower}). Any $follower$ robot, after observing $increase=true$ in a 2-dedicated round, moves through the port $child$ and comes back to its current node. Hence, the total number of 2-dedicated round used is equal to the length of the path $P$. Since exactly one robot leaves $s$ in every iteration, there can be at most $j-2$ $follower$ robot and one $master$ robot. Hence, the length of the path $P$ can be at most $j-1$.

Since there are $k$ iterations of the algorithm, the total number of 2-dedicated rounds used is $O(k^2)$.
\end{proof}

\begin{lemma}\label{lem:4round}
The total number of 4-dedicated round is $O(k^2 \log \Delta)$.
\end{lemma}
\begin{proof}
We first show that, for $1 \le j \le k-1$, the $j$-th iteration uses $O(j \log \Delta)$ 4-dedicated rounds.

In any iteration, the 4-dedicated rounds are used to execute the subroutine \ssl~by a $master$ or $follower$ robot. According to
subroutine \ssl, (Algorithm \ref{alg:sendsignal}), the robot moves to its parent node in a 4-dedicated round only if the corresponding bit of $\alpha$ is 1. Note that the first two bits of $\alpha$ are one of  11, 10 and  01 and the next bits represents transformed binary encoding of the port $child$ of the current node. Hence the length of $\alpha$ is $O(\log \Delta)$. Since the movement of the robot depends only on the length of $\alpha$, the robot uses $O(\log \Delta)$ 4-dedicated rounds. Also, since the number of $master$ or $follower$ robot in the $j$-th iteration is at most $j-1$, hence the total number of 4-dedicated rounds in an iteration is $O(j \log \Delta)$.

Since there are $k$ iterations, the total number of 4-dedicated rounds is $O(k^2 \log \Delta)$.
\end{proof}

\begin{lemma}\label{lem:5round}
The total number of 5-dedicated round is $O(k^2)$.

\end{lemma}
\begin{proof}
We first show that, for $1 \le j \le k-1$, the $j$-th iteration uses $O(j)$ 5-dedicated rounds. According to Algorithm \ref{algo:master} and Algorithm \ref{algo:follower}, the $master$ and each $follower$ robot uses at most one 5-dedicated round. Since the number of $master$ or $follower$ robot in the $j$-th iteration is at most $j-1$, hence the total number of 5-dedicated rounds in an iteration is $O(j)$. Since there are $k$ iterations, the total number of 5-dedicated rounds is $O(k^2)$.
\end{proof}

\begin{lemma}\label{lem:3round}
The total number of 3-dedicated round  used across all the iterations is $O(\min\{k\Delta,k^2\})$.
\end{lemma}
\begin{proof}
The 3-dedicated rounds are used only by a $master$ $robot$ to search for an empty neighbor.
Consider any node $v$ in $G$ which is full at the end of the $k$-th iteration. Since a $master$ robot searches each port of a node at most once, and there is at most $k$ nodes from where the $master$ robot searches for empty neighbors, therefore, the total number of 3-dedicated rounds used is $O(k \Delta)$.

Now consider the executions of 3-dedicated rounds by a $master$ robot. As the total number of robots is $k$, a $master$ robot can not find more than $k$ full neighbor from any node. Also, total number of nodes from where the $master$ robot searches for empty neighbors is at most $k$. Hence, the total number of 3-dedicated rounds used is $O(k^2)$.

Hence, the total number of 3-dedicated rounds is $O(\min\{k\Delta,k^2\})$.
\end{proof}

\begin{lemma}
\label{lem:0round}
The total number of $0$-dedicated round is at most $k$.
\end{lemma}
\begin{proof}
The 0-dedicated rounds are used only in the last iteration. According to Algorithm \ref{algo:master} (steps 5-6) and Algorithm \ref{algo:follower} (steps 10-12), each robot uses at most one 0-dedicated round and then becomes $idle$. Since the total number of $master$ or $follower$ robot in the $k$-th iteration is at most $k$, the  total number of $0$-dedicated round is at most $k$.
\end{proof}
\begin{theorem}
The algorithm terminates in time $O(k \log L+ k^2 \log \Delta)$ and each robot uses $O(\log L+ \log \Delta)$ additional memory.
\end{theorem}
\begin{proof}
The time complexity of the algorithm is asymptotically bounded above by the maximum number of required $j$-dedicated rounds, where $0 \le j \le 5$. Since for any two functions $f,g$, $\max\{f,g\} \in O(f+g)$, the time complexity of the algorithm is obtained from Lemma \ref{lem:1round} to Lemma \ref{lem:0round} and Theorem \ref{theorem:kiteration}, is equals to  $O(k \log L+ k^2 \log \Delta+ (min\{k\Delta,k^2\}))$. Since $min\{k\Delta,k^2\}) \in O(k^2 \log \Delta)$, hence the time complexity of our algorithm is
$O(k \log L+ k^2 \log \Delta)$.

The variable $j$ which is used as a parameter in the subroutine \pb~uses $O(\log L)$ memory. The variables $\alpha',\alpha$ (used in the subroutine \lsl), $child$, $parent$, $q,q'$ (used to store the incoming and outgoing  ports in the subroutine \pb), $prt$ (used to store the port number while searching for empty neighbor in the subroutine \master $(M)$) uses memory of size $O(\log \Delta)$. All the other variables uses by the robots are of constant memory. Hence, the additional memory used by each robot is $O(\log L+\log \Delta)$.
\end{proof}

We shows that the amount of additional memory used by every robot in the algorithm is indeed asymptotically optimal. In \cite{MollaM19}, the authors proved $\Omega(\log \Delta)$ lower bound of memory requirement by any randomized algorithm for each robot for dispersion.  The same proof gives $\Omega(\log \Delta)$ lower bound of memory for any deterministic algorithm.
Hence it is enough to prove the lower bound $\Omega(\log L)$ which we do in the following theorem.

\begin{theorem}
In the proposed communication model, dispersion can not be achieved by a set of mobile robots if the memory available to the robots is $o(\log L)$.
\end{theorem}
\begin{proof}
Let $\cA$ be an algorithm using which a set of $k$ robots starting from a node of any graph $G$ can achieve dispersion and terminate in finite rounds, where the robots uses at most  $\frac{\log L}{2}$ bits of memory.

Define {\it state} of a robot in a round $t$, as the snapshot of the specified memory of the robot at the beginning of the round $t$. For example, if a robot has 5 bits of memory and 11100 is stored in the memory, then the $state=11100$. In any round, the decision of a robot at a node $v$ depends only on its $state$ in that round and the occurrences of the events at $v$. At the beginning of round 1, there are no events occurred in $v$ and the only information available to a robot is its id stored in its fixed memory.

We construct $\frac{L(L-1)}{2}$ different inputs as follows. For $0 \le i,j \le L-1$, the input  $<G,i,j>$, where the graph $G$ is a graph with two nodes connected by an edge, and two robots with ids $i,j$ are placed one one node of $G$.

Since at the beginning of round 1, the state of the robot depends only on its id, the robot with id $i$ have the same state in all the inputs $<G,i,j>$, $0 \le j\le L-1$, and $i \ne j$.

Now, since the size of the memory of each robot is at most $\frac{\log L}{2}$, there are at most $2^{\frac{\log L}{2}+1}< L$ different states at the beginning of round 1. Since there are $L$ robots among all the inputs are present with different ids, by Pigeonhole principle, there must exist two robots with different ids that have the same states. Let these ids are $i',j'$.

Consider the execution of the algorithm $\cA$ for the input $<G,i',j'>$. Here, the two robots have the same state at the beginning of round 1.
We claim that in any subsequent round, they have the same states and they remain co-located.
Since they experiences same events, if any, happening in this node in round 1, their decisions are going to be the same. Accordingly, either both of them move or both of them stay in the current node. Let they remain together till the beginning of the  $t$-th round and have the same states.  As in the $t$-th round, they observe the same event, so either both of them move or both of stay in the current node. So in beginning of the $t+1$-th round, they remain together. As a result, they never get separated. Which contradicts the fact the $\cA$ achieves dispersion on any graph with at most $\frac{\log L}{2}$ bits of memory.
\end{proof}

\section{Conclusion}
This paper introduces an algorithm that achieved dispersion without any communication between the robots using asymptotically optimal additional memory. Here, the task of dispersion is achieved under the assumption that the robots have access to two local information at any node: (1) whether the robot is alone at the node (2) whether the number of robots changes at the node compared to the previous round. A natural question arises that whether dispersion can be achieved with lesser local information as well. To be specific, it will be quite interesting to study whether the information of a robot is alone or not at a node is sufficient to achieve dispersion. Also, improving the time complexity of our algorithm in the proposed model or proving a lower bound of the same is another problem which can be explored in the future.
\bibliographystyle{IEEEtran}
\bibliography{biblio}

\begin{thebibliography}{10}

\bibitem{AgarwallaAMKS18}
Ankush Agarwalla, John Augustine, William K.~Moses Jr., Sankar~Madhav K., and
  Arvind~Krishna Sridhar.
\newblock Deterministic dispersion of mobile robots in dynamic rings.
\newblock In {\em {ICDCN}}, pages 19:1--19:4, 2018.

\bibitem{Alpern2003}
Steve Alpern and Shmuel Gal.
\newblock {\em The theory of search games and rendezvous}, volume~55 of {\em
  International series in operations research and management science}.
\newblock Kluwer, 2003.

\bibitem{AugustineM18}
John Augustine and William K.~Moses Jr.
\newblock Dispersion of mobile robots: {A} study of memory-time trade-offs.
\newblock In {\em {ICDCN}}, pages 1:1--1:10, 2018.

\bibitem{BarriereFBS11}
Lali Barri{\`{e}}re, Paola Flocchini, Eduardo~Mesa Barrameda, and Nicola
  Santoro.
\newblock Uniform scattering of autonomous mobile robots in a grid.
\newblock {\em Int. J. Found. Comput. Sci.}, 22(3):679--697, 2011.

\bibitem{BouchardDP20}
S{\'{e}}bastien Bouchard, Yoann Dieudonn{\'{e}}, and Andrzej Pelc.
\newblock Want to gather? no need to chatter!
\newblock In {\em {PODC}}, pages 253--262, 2020.

\bibitem{BrassCGX11}
Peter Brass, Flavio Cabrera{-}Mora, Andrea Gasparri, and Jizhong Xiao.
\newblock Multirobot tree and graph exploration.
\newblock {\em {IEEE} Trans. Robotics}, 27(4):707--717, 2011.

\bibitem{BrassVX14}
Peter Brass, Ivo Vigan, and Ning Xu.
\newblock Improved analysis of a multirobot graph exploration strategy.
\newblock In {\em {ICARCV}}, pages 1906--1910, 2014.

\bibitem{Das2020}
Archak Das, Kaustav Bose, and Buddhadeb Sau.
\newblock Memory optimal dispersion by anonymous mobile robots.
\newblock In {\em CALDAM}, pages 426--439, 2021.

\bibitem{DereniowskiDKPU15}
Dariusz Dereniowski, Yann Disser, Adrian Kosowski, Dominik Pajak, and
  Przemyslaw Uznanski.
\newblock Fast collaborative graph exploration.
\newblock {\em Inf. Comput.}, 243:37--49, 2015.

\bibitem{ElorB11}
Yotam Elor and Alfred~M. Bruckstein.
\newblock Uniform multi-agent deployment on a ring.
\newblock {\em Theor. Comput. Sci.}, 412(8-10):783--795, 2011.

\bibitem{KshemkalyaniF19}
Ajay~D. Kshemkalyani and Faizan Ali.
\newblock Efficient dispersion of mobile robots on graphs.
\newblock In {\em ICDCN}, page 218–227, 2019.

\bibitem{KshemkalyaniMS19}
Ajay~D. Kshemkalyani, Anisur~Rahaman Molla, and Gokarna Sharma.
\newblock Fast dispersion of mobile robots on arbitrary graphs.
\newblock In {\em {ALGOSENSORS}}, pages 23--40, 2019.

\bibitem{KshemkalyaniMS20walcom}
Ajay~D. Kshemkalyani, Anisur~Rahaman Molla, and Gokarna Sharma.
\newblock Dispersion of mobile robots on grids.
\newblock In {\em {WALCOM}}, pages 183--197, 2020.

\bibitem{Kshemkalyanijpdc22}
Ajay~D. Kshemkalyani, Anisur~Rahaman Molla, and Gokarna Sharma.
\newblock Dispersion of mobile robots using global communication.
\newblock {\em J. Parallel Distributed Comput.}, 161:100--117, 2022.

\bibitem{Kshemkalyaniarxiv}
Ajay~D. Kshemkalyani and Gokarna Sharma.
\newblock Near-optimal dispersion on arbitrary anonymous graphs.
\newblock {\em CoRR}, abs/2106.03943, 2021.

\bibitem{MollaM19}
Anisur~Rahaman Molla and William K.~Moses Jr.
\newblock Dispersion of mobile robots: The power of randomness.
\newblock In {\em TAMC}, pages 481--500, 2019.

\bibitem{MollaMM20}
Anisur~Rahaman Molla, Kaushik Mondal, and William K.~Moses Jr.
\newblock Efficient dispersion on an anonymous ring in the presence of weak
  byzantine robots.
\newblock In {\em {ALGOSENSORS}}, pages 154--169, 2020.

\bibitem{Anisur21}
Anisur~Rahaman Molla, Kaushik Mondal, and William K.~Moses Jr.
\newblock Byzantine dispersion on graphs.
\newblock In {\em IPDPS}, pages 942--951, 2021.

\bibitem{Pelc12}
Andrzej Pelc.
\newblock Deterministic rendezvous in networks: {A} comprehensive survey.
\newblock {\em Networks}, 59(3):331--347, 2012.

\bibitem{ShibataMOKM16}
Masahiro Shibata, Toshiya Mega, Fukuhito Ooshita, Hirotsugu Kakugawa, and
  Toshimitsu Masuzawa.
\newblock Uniform deployment of mobile agents in asynchronous rings.
\newblock In {\em {PODC}}, pages 415--424, 2016.

\bibitem{ShintakuSKM20}
Takahiro Shintaku, Yuichi Sudo, Hirotsugu Kakugawa, and Toshimitsu Masuzawa.
\newblock Efficient dispersion of mobile agents without global knowledge.
\newblock In {\em {SSS}}, pages 280--294, 2020.

\end{thebibliography}
\end{document}